\documentclass[letterpaper, 10 pt, conference]{ieeeconf}  % Comment this line out if you need a4paper

\IEEEoverridecommandlockouts                              % This command is only needed if 
\usepackage{mathtools}
\usepackage{epsfig} % for postscript graphics files
\usepackage{subfigure}
\usepackage{times} % assumes new font selection scheme installed
\usepackage{amsmath} % assumes amsmath package installed
\usepackage{amssymb}  % assumes amsmath package installed
\usepackage{pdfsync}
\usepackage{enumerate}
\usepackage{color}
\usepackage{lineno}
\usepackage{amsmath,bm}
\usepackage{cite}

\usepackage{dsfont} %for \mathds{1}

\usepackage{graphicx}
	\graphicspath{{figures/}}

\usepackage{tikz}
	\usetikzlibrary{calc}
	\usetikzlibrary{shapes,arrows}
	\tikzstyle{frame} = [draw, -latex]
	\tikzstyle{lineUD} = [draw]
	\tikzstyle{line} = [draw, -latex']
	\tikzstyle{line2} = [draw, -latex', dashdotted]
	\tikzstyle{line3} = [draw, -latex', dashed]
	\tikzstyle{line3UD} = [draw, dashed]
	\tikzstyle{place} = [circle, draw=black, fill=white, thick, inner sep=2pt, minimum size=1mm]
	\tikzstyle{placeRed} = [circle, draw=red, fill=red, thick, inner sep=2pt, minimum size=1mm]
	\tikzstyle{vertex} = [circle, draw=black, fill=black, thick, inner sep=2pt, minimum size=1mm]
\usepackage{tkz-euclide}

\newtheorem{definition}{Definition}[section]
\newtheorem{problem}{Problem}[section]

\newtheorem{lemma}{Lemma}[section]
\newtheorem{theorem}{Theorem}[section]
\newtheorem{corollary}{Corollary}[section]

\newcommand{\mbf}{\mathbf}

\newcommand{\myemph}{\emph}
\newcommand{\myfig}{Fig.~}
\DeclareMathOperator{\rank}{{rank}}

\DeclarePairedDelimiter \card{\lvert}{\rvert}
\DeclarePairedDelimiter \norm{\|}{\|}

\providecommand\given{}
\newcommand\SetSymbol[1][]{\nonscript\;#1\vert\nonscript\;
\mathopen{}\allowbreak}
\DeclarePairedDelimiterX\set[1]\{\}{%
\renewcommand\given{\SetSymbol[\delimsize]}
#1
}

\title{\LARGE \bf
Infinitesimal Weak Rigidity, Formation Control of Three Agents, and Extension to $3$-dimensional Space
}
\author{Seong-Ho Kwon${}^{\dagger}$, Minh Hoang Trinh${}^{\dagger}$, Koog-Hwan Oh${}^{\dagger}$, Shiyu Zhao${}^{\ddagger}$, and Hyo-Sung Ahn${}^{\dagger}$% <-this % stops a space
\thanks{$^{\dagger}$School of  Mechanical Engineering, Gwangju Institute of Science and Technology (GIST), Gwangju, Korea.
        {\tt\small \{seongho, trinhhoangminh, ohkhwan, hyosung\}@gist.ac.kr}
	}
\thanks{$^{\ddagger}$Department of Automatic Control and Systems Engineering, University of Sheffield, UK.
		{\tt\small szhao@sheffield.ac.uk}
	}
}

\begin{document}
%\linenumbers

\maketitle
\thispagestyle{empty}
\pagestyle{empty}

%%%%%%%%%%%%%%%%%%%%%%%%%%%%%%%%%%%%%%%%%%%%%%%%%%%%%%%%%%%%%%%%%%%%%%%%%%%%%%%%
\begin{abstract}
In this paper, we introduce new concepts of weak rigidity matrix and infinitesimal weak rigidity for planar frameworks. The weak rigidity matrix is used to directly check if a framework is infinitesimally weakly rigid while previous work can check a weak rigidity of a framework indirectly. An infinitesimal weak rigidity framework can be uniquely determined up to a translation and a rotation (and a scaling also when the framework does not include any edge) by its inter-neighbor distances and angles. We apply the new concepts to a three-agent formation control problem with a gradient control law, and prove instability of the control system at any incorrect equilibrium point and convergence to a desired target formation. Also, we propose a modified Henneberg construction, which is a technique to generate minimally rigid (or weakly rigid) graphs.  Finally, we extend the concept of the weak rigidity in $\mathbb{R}^{2}$ to the concept in $\mathbb{R}^{3}$. 
\end{abstract}
%%%%%%%%%%%%%%%%%%%%%%%%%%%%%%%%%%%%%%%%%%%%%%%%%%%%%%%%%%%%%%%%%%%%%%%%%%%%%%%%
\section{INTRODUCTION}
\label{Sec:intro}

Rigid formation shape is an important requirement in many formation control and network localization problems. Specific or fixed formation shape may be useful for sensing agents, localizing agents, moving agents from one location to another and moving objects. A lot of control methods to achieve a target formation shape have been reported in the literature \cite{oh2015survey,krick2009stabilisation,helmke2014geometrical,sun2015rigid,zhao2016bearing}. One of the formation control methods is \myemph{distance-constrained (distance-based)} formation control \cite{krick2009stabilisation,helmke2014geometrical}, where the target formation is achieved by ​obtaining the inter-agent distances. Another one is \myemph{bearing-constrained (bearing-based) formation control} \cite{zhao2016bearing,bishop2011stabilization} where the target formation is achieved by obtaining the inter-agent bearings. Also, there is a mixed method of distance and bearing constrained formation control \cite{bishop2015distributed}. Another one is to make use of only relative angles \cite{buckley2017infinitesimally} where maintains the target formation by sensing relative angle measurements.

In the distance-constrained formation control problem, one approach to characterize a unique formation shape (at least locally) is the (distance) \myemph{rigidity} of a framework \cite{asimow1979rigidity}. In the bearing-constrained formation control, the theory to characterize unique formation shape is the \myemph{bearing rigidity} of a framework \cite{zhao2016bearing,zelazo2014rigidity}. In a mixed method of distance and bearing constrained formation control, there is no specific rigidity theory. In \cite{bishop2015distributed}, the authors developed a control law using inter-agent bearing and distance constraints. Recently, in particular, the only angle constrained formation control \cite{buckley2017infinitesimally} and new rigidity theory with distance and subtended-angle constraints, named \myemph{weak rigidity} \cite{park2017rigidity}, were introduced. In \cite{buckley2017infinitesimally}, they make use of a shape-similarity matrix to preserve a formation shape by only using relative angle measurements. If the null space of the shape-similarity matrix includes trivial motions only up to a translation, a rotation and a scaling, then the formation shape is preserved. This concept is similar to the (distance) rigidity and bearing rigidity. In \cite{park2017rigidity}, a formation shape ​whose​ shape can be (locally) uniquely determined specified by inter-agent distance and subtended-angle constraints is considered to be weakly rigid even though it is non-rigid in the viewpoint of (distance) rigidity. However, whether the formation is weakly rigid cannot be determined directly from the original framework. The method proposed in \cite{park2017rigidity} requires to transform the original framework into another framework with distance-only constraints. Then, if this transformed framework is  rigid, we can conclude that the original framework is weakly rigid. Thus, it is inconvenient to check the weak rigidity based on the proposed method in \cite{park2017rigidity}.

In this paper, our main contributions are summarized as follows. First, we provide new concepts of \myemph{weak rigidity matrix} and \myemph{infinitesimal weak rigidity} in the two-dimensional space. For a given framework in $\mathbb{R}^{2}$, we propose a method to construct a corresponding weak rigidity matrix from the set of mixed distance- and angle-contraints. The rank of the weak rigidity matrix can be used to check infinitesimal weak rigidity of the framework. A framework defined by a set of mixed distance- and angle-constraints is infinitesimally weakly rigid if the null space of its weak rigidity matrix is spanned by only rigid body translations and rotations. Moreover, if an infinitesimally weakly rigid framework is specified by only some angle constraints, the null space of the weak rigidity matrix contains also scalings. As a result, the existing distance rigidity and bearing rigidity theories in the literature could be unified into the weak rigidity theory.
Second, we apply the concept of the infinitesimal weak rigidity to a formation control with three agents in the two-dimensional space. We prove that the three-agent formation at any incorrect equilibrium is unstable by investigating the eigenvalues of the Jacobian of the formation system. We prove that the system converges to a desired target formation from almost global initial positions. Also, we introduce a modified Henneberg construction using an angle extension. The construction is used to grow minimally rigid formations, which are useful in designing a formation control strategy \cite{mou2015target,trinh2016bearing}. Finally, we extend the concept of the weak rigidity \cite{park2017rigidity} in the two-dimensional space to the concept in the three-dimensional space. 

The rest of this paper is organized as follows. Section \ref{Sec:weakRigidity} briefly reviews the background of the weak rigidity in $\mathbb{R}^{2}$. Section \ref{Sec:Infinitesimally_weakRigidity} provides the new concepts of the weak rigidity matrix and infinitesimal weak rigidity. The relation between infinitesimal weak rigidity and the rank of the weak rigidity matrix is also established. In Section \ref{Sec:Formation Control Problem}, we provide the analysis of the instability of incorrect equilibria and the convergence result of a three-agent formation system. In Section \ref{Sec:MHenneberg}, we discuss and define the modified Henneberg construction. In Section \ref{Sec:weakRigidity_3dim}, the weak rigidity is extended from the two-dimensional space to the concept in the three-dimensional space.
Lastly, conclusion and summary are provided in Section \ref{Sec:conclusion}.

\myemph{Preliminaries and Notations}: The notation $\norm{\cdot}$ means the Euclidean norm of a vector and the notation $\card{\mathcal{S}}$ means the cardinality of a set $\mathcal{S}$. 
Let $K_n$ denote a complete graph with $n$ vertices s.t. $K_n = (\mathcal{V}_K,\mathcal{E}_K)$, then an undirected graph $\mathcal{G}$ is defined as $\mathcal{G} = (\mathcal{V},\mathcal{E},\mathcal{A})$, where a vertex set $\mathcal{V}=\set{1,2,...,n}$, an edge set $\mathcal{E} \subseteq \mathcal{V} \times \mathcal{V}$ with $m=\card{\mathcal{E}}$ and an angle set $\mathcal{A} = \set{(k,i,j) \given \theta_{ij}^{k} \text{ is assigned to } (i,k), (j,k) \in \mathcal{E}_K, \theta_{ij}^{k} \in [0,\pi]}$ with $q=\card{\mathcal{A}}$. We assume that duplicated edges between any two vertices do not exist, e.g., $(i,j) = (j,i)$ for all $i,j \in \mathcal{V}$.
The $\theta_{ij}^{k}$ means the angle subtended by the adjacent edges $(i,k)$ and $(j,k)$. The set of neighbors of vetex $i$ is denoted as $\mathcal{N}_i = \set{j \in \mathcal{V} \given (i,j) \in \mathcal{E}}$. For a position vector $\mbf{p}_i \in \mathbb{R}^{2}$, a configuration of $\mathcal{G}$ in $\mathbb{R}^{2}$ is defined as $\mbf{p} = [\mbf{p}_{1}^\top,...,\mbf{p}_{n}^\top]^\top \in \mathbb{R}^{2n}$, and a framework is defined as $(\mathcal{G},\mbf{p})$. Two frameworks $(\mathcal{G},\mbf{p})$ and $(\mathcal{G},\mbf{q})$ are said to be \myemph{congruent} if $\norm{\mbf{p}_{i}-\mbf{p}_{j}}=\norm{\mbf{q}_{i}-\mbf{q}_{j}}$ for all $i,j \in \mathcal{V}$. Also, two frameworks $(\mathcal{G},\mbf{p})$ and $(\mathcal{G},\mbf{q})$ are said to be \myemph{equivalent} if $\norm{\mbf{p}_{i}-\mbf{p}_{j}}=\norm{\mbf{q}_{i}-\mbf{q}_{j}}$ for all $(i,j) \in \mathcal{E}$. For a framework $(\mathcal{G},\mbf{p})$, the relative position vector and the relative distance are defined as $\mbf{z}_{ij} = \mbf{p}_{i} - \mbf{p}_{j}$  and $d_{ij} = \norm{\mbf{z}_{ij}}$, respectively, for all $(i,j)\in \mathcal{E}$.
Let Null$(\cdot)$ and rank$(\cdot)$ be the null space and the rank of a matrix, respectively. Denote $I_N \in \mathbb{R}^{N \times N}$ as an identity matrix, and $\mathds{1} = [1, ..., 1]^\top$. The perpendicular operator $J \in \mathbb{R}^{2 \times 2}$ is denoted as $J \triangleq \begin{bmatrix}
0 &-1\\
1 & 0
\end{bmatrix}.$
We assume that i) there is no $\myemph{self-loop}$, i.e. $(i,i) \notin \mathcal{E}$ for any vertex $i \in \mathcal{V}$, ii) formations are undirected, and iii) there are no position vectors collocated at one point. %there are no position vectors corresponding to each other 는 weak rigidity matrix 가 정의가 안되는 경우(cosine이 정의가 안되는 경우) 때문에 넣음.
%$\mathcal{A} = \set*{(k,i,j) \given \text{$\theta_{ij}^{k}$ is assigned to two edges \((i,k),  (j,k) \in \mathcal{E}\)}, \theta_{ij}^{k} \in [0,\pi]}$ 
%%%%%%%%%%%%%%%%%%%%%%%%%%%%%%%%%%%%%%

%%%%%%%%%%%%%%%%%%%%%%%%%%%%%%%%%%%%%%%%%%%%%%%%%%%%%%%%%%%%%%%%%%%%%%%%%%%%%%%%
\section{Background of the Weak Rigidity Theory}%\section{Weak Rigidity in the Two-Dimensional Space}
\label{Sec:weakRigidity}
%%%%%%%%%%%%%%%%%%%%%%%%%%%%%%
%\subsection{Weak Rigidity} \label{Subsec:weakRigidity}
In this section, we briefly review the concepts of the weak rigidity in \cite{park2017rigidity}. The weak rigidity theory is concerned with frameworks defined by distance constraints and additional subtended-angle constraints in $\mathbb{R}^2$. Distance constraints and additional subtended-angle constraints are required to achieve a unique formation shape under the weak rigidity theory. 

\begin{definition}
\label{Def:strongEquiv}
With $n \ge 3$, two frameworks $(\mathcal{G},\mbf{p})$ and $(\mathcal{G},\mbf{q})$ are said to be \myemph{strongly equivalent} if the following two conditions hold:
\begin{itemize}
\item $\norm{\mbf{p}_{v}-\mbf{p}_{w}} = \norm{\mbf{q}_{v}-\mbf{q}_{w}}, \forall (v,w) \in \mathcal{E}$,
\item ${\theta_{ij}^{k}}_{\in(\mathcal{G},\mbf{p})} = {\theta_{ij}^{k}}_{\in(\mathcal{G},\mbf{q})}, \forall (k,i,j) \in \mathcal{A}$,
\end{itemize}
\end{definition}
where ${\theta_{ij}^{k}}_{\in(\mathcal{G},\mbf{p})}$ and ${\theta_{ij}^{k}}_{\in(\mathcal{G},\mbf{q})}$ denote the subtended angles  in $(\mathcal{G},\mbf{p})$ and $(\mathcal{G},\mbf{q})$, respectively.

\begin{definition}
\label{Def:weakRigidity}
A framework $(\mathcal{G},\mbf{p})$ is \myemph{weakly rigid} in $\mathbb{R}^{2}$ if there exists a neighborhood $\mathcal{B}_{\mbf{p}} \subseteq \mathbb{R}^{2n}$ of $\mbf{p}$ such that each framework $(\mathcal{G},\mbf{q})$, $\mbf{q} \in \mathcal{B}_{\mbf{p}}$, strongly equivalent to $(\mathcal{G},\mbf{p})$ is congruent to $(\mathcal{G},\mbf{p})$.
\end{definition}

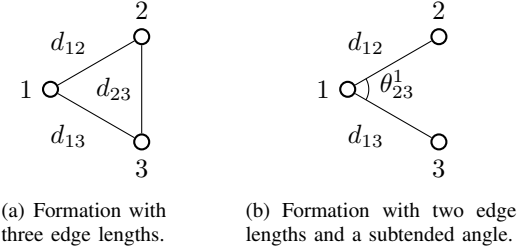
\begin{figure}[t]\setcounter{subfigure}{0}
\centering
\subfigure[Formation with three edge lengths.]{ \label{Formation_3e}
\begin{tikzpicture}[scale=0.7]
\node[place] (node1) at (-1.732,0) [label=left:$1$] {};
\node[place] (node2) at (0,1) [label=above:$2$] {};
\node[place] (node3) at (0,-1) [label=below:$3$] {};

\draw[lineUD] (node1)  -- node [above left] {$d_{12}$} (node2);
\draw[lineUD] (node1)  -- node [below left] {$d_{13}$} (node3);
\draw[lineUD] (node2)  -- node [left] {$d_{23}$} (node3);
\end{tikzpicture}
}%1
\qquad\quad %space function
\subfigure[Formation with two edge lengths and a subtended angle.]{ \label{Formation_3e1a}
\qquad\begin{tikzpicture}[scale=0.7]
\node[place] (node1) at (-1.732,0) [label=left:$1$] {};
\node[place] (node2) at (0,1) [label=above:$2$] {};
\node[place] (node3) at (0,-1) [label=below:$3$] {};

\draw[lineUD] (node1)  -- node [above left] {$d_{12}$} (node2);
\draw[lineUD] (node1)  -- node [below left] {$d_{13}$} (node3);

\draw [black, domain=330:390] plot ({0.4*cos(\x) - 1.732}, {0.4*sin(\x)}) node at (-1.3,0.1) [right] {$\theta_{23}^{1}$};
\end{tikzpicture} \qquad 
}%2
\caption{\small Two different but congruent triangular formations. \normalsize} \label{Triangular_formations}
\end{figure}

Two congruent frameworks are illustrated in Fig. \ref{Triangular_formations}. Fig. \ref{Formation_3e} is defined by three edge lengths while the other in Fig. \ref{Formation_3e1a} is defined by two edge lengths and a subtended angle with the condition $({d}_{23})^{2} = ({d}_{12})^{2} + ({d}_{13})^{2} - 2{d}_{12}{d}_{13}\cos\theta_{23}^{1}$ induced from the law of cosines. The two formations can be changed to each other with the condition induced by the law of cosines. That is, either three distance constraints or two distance constraints with a subtended angle can define the same triangular formation. 

%%%%%%%%%%%%%%%%%%%%%%%%%%%%%%%%%%%%%%%%%%%%%%%%%%%%%%%%%%%%%%%%%%%%%%%%%%%%%%%%
\section{Infinitesimal Weak Rigidity} %in the 2-Dimensional Space} 
\label{Sec:Infinitesimally_weakRigidity}
%%%%%%%%%%%%%%%%%%%%%%%%%%%%%%
In this section, we introduce the weak rigidity matrix and infinitesimal weak rigidity, and provide a rank condition of the weak rigidity matrix to determine if a framework is infinitesimally weakly rigid in $\mathbb{R}^{2}$ in a straightforward way. In \cite{park2017rigidity}, an angle $\theta_{ij}^{k}$ must be defined with adjacent two edges, i.e. $(i,k)$, $(j,k)\in \mathcal{E}$. However, with the weak rigidity matrix, the adjacent edges do not need to be defined. For example, we can check whether a framework with only angle constraints is infinitesimally weakly rigid or not by a rank condition of weak rigidity matrix.

%%%%%%%%%%%%%%%%%%%%%%%%%%%%%%
\subsection{Weak Rigidity Matrix}

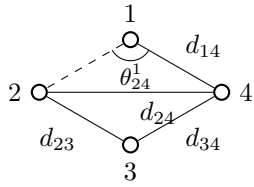
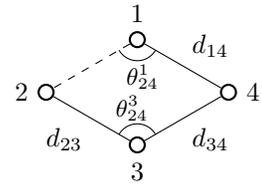
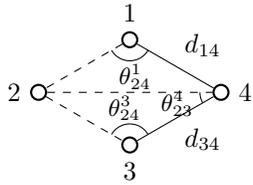
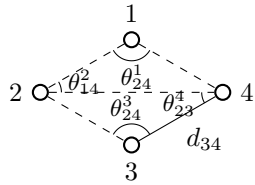
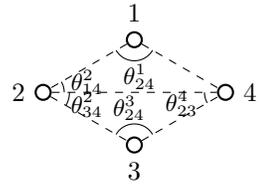
\begin{figure*}[b]\setcounter{subfigure}{0}
\centering
\subfigure[Rigid and infinitesimally weakly rigid framework.]{ \label{inf_ex01}
\begin{tikzpicture}[scale=0.7]
\node[place] (node1) at (0,1) [label=above:$1$] {};
\node[place] (node2) at (-1.732,0) [label=left:$2$] {};
\node[place] (node3) at (0,-1) [label=below:$3$] {};
\node[place] (node4) at (1.732,0) [label=right:$4$] {};

\draw[lineUD] (node1)  -- node [above left] {$d_{12}$} (node2);
\draw[lineUD] (node1)  -- node [above right] {$d_{14}$} (node4);
\draw[lineUD] (node2)  -- node [below left] {$d_{23}$} (node3);
\draw[lineUD] (node3)  -- node [below right] {$d_{34}$} (node4);
\draw[lineUD] (node2)  -- node [above] {$d_{24}$} (node4);
\end{tikzpicture} \qquad 
}\quad\quad\quad%1
\subfigure[Infinitesimally weakly rigid framework with constraints of 4 edges and an angle.]{ \label{inf_ex02}
\qquad\begin{tikzpicture}[scale=0.7]
\node[place] (node1) at (0,1) [label=above:$1$] {};
\node[place] (node2) at (-1.732,0) [label=left:$2$] {};
\node[place] (node3) at (0,-1) [label=below:$3$] {};
\node[place] (node4) at (1.732,0) [label=right:$4$] {};

\draw[dashed] (node1)  -- (node2);
\draw[lineUD] (node1)  -- node [above right] {$d_{14}$} (node4);
\draw[lineUD] (node2)  -- node [below left] {$d_{23}$} (node3);
\draw[lineUD] (node3)  -- node [below right] {$d_{34}$} (node4);
\draw[lineUD] (node2)  -- node [below right] {$d_{24}$} (node4);
\draw [black, domain=210:330] plot ({0.4*cos(\x)}, {0.4*sin(\x) + 1}) node at (0.1,0.3) {\small$\theta_{24}^{1}$};
\end{tikzpicture} \qquad
}\quad\quad\quad%2
\subfigure[Infinitesimally weakly rigid framework with constraints of 3 edges and 2 angles.]{ \label{inf_ex03}
\qquad\begin{tikzpicture}[scale=0.7]
\node[place] (node1) at (0,1) [label=above:$1$] {};
\node[place] (node2) at (-1.732,0) [label=left:$2$] {};
\node[place] (node3) at (0,-1) [label=below:$3$] {};
\node[place] (node4) at (1.732,0) [label=right:$4$] {};

\draw[dashed] (node1)  -- (node2);
\draw[lineUD] (node1)  -- node [above right] {$d_{14}$} (node4);
\draw[lineUD] (node2)  -- node [below left] {$d_{23}$} (node3);
\draw[lineUD] (node3)  -- node [below right] {$d_{34}$} (node4);
\draw [black, domain=210:330] plot ({0.4*cos(\x)}, {0.4*sin(\x) + 1}) node at (0.1,0.3) {\small$\theta_{24}^{1}$};
\draw [black, domain=30:150] plot ({0.4*cos(\x)}, {0.4*sin(\x) - 1}) node at (-0.1,-0.3) {\small$\theta_{24}^{3}$};
\end{tikzpicture} \qquad 
}\newline %3 
\subfigure[Infinitesimally weakly rigid framework with constraints of 2 edges and 3 angles.]{ \label{inf_ex04}
\begin{tikzpicture}[scale=0.7]
\node[place] (node1) at (0,1) [label=above:$1$] {};
\node[place] (node2) at (-1.732,0) [label=left:$2$] {};
\node[place] (node3) at (0,-1) [label=below:$3$] {};
\node[place] (node4) at (1.732,0) [label=right:$4$] {};

\draw[dashed] (node1)  -- (node2);
\draw[lineUD] (node1)  -- node [above right] {$d_{14}$} (node4);
\draw[dashed] (node2)  -- (node3);
\draw[lineUD] (node3)  -- node [below right] {$d_{34}$} (node4);
\draw[dashed] (node2)  -- (node4);
\draw [black, domain=210:330] plot ({0.4*cos(\x)}, {0.4*sin(\x) + 1}) node at (0.1,0.3) {\small$\theta_{24}^{1}$};
\draw [black, domain=30:150] plot ({0.4*cos(\x)}, {0.4*sin(\x) - 1}) node at (-0.1,-0.3) {\small$\theta_{24}^{3}$};
\draw [black, domain=180:210] plot ({0.4*cos(\x)+1.732}, {0.4*sin(\x)}) node at (0.9,-0.21) {\small$\theta_{23}^{4}$};
\end{tikzpicture} \qquad \qquad
}\quad\quad\quad%4
\subfigure[Infinitesimally weakly rigid framework with constraints of an edge and 4 angles.]{ \label{inf_ex05}
\begin{tikzpicture}[scale=0.7]
\node[place] (node1) at (0,1) [label=above:$1$] {};
\node[place] (node2) at (-1.732,0) [label=left:$2$] {};
\node[place] (node3) at (0,-1) [label=below:$3$] {};
\node[place] (node4) at (1.732,0) [label=right:$4$] {};

\draw[dashed] (node1)  -- (node2);
\draw[dashed] (node1)  -- (node4);
\draw[dashed] (node2)  -- (node3);
\draw[lineUD] (node3)  -- node [below right] {$d_{34}$} (node4);
\draw[dashed] (node2)  -- (node4);
\draw [black, domain=210:330] plot ({0.4*cos(\x)}, {0.4*sin(\x) + 1}) node at (0.1,0.3) {\small$\theta_{24}^{1}$};
\draw [black, domain=30:150] plot ({0.4*cos(\x)}, {0.4*sin(\x) - 1}) node at (-0.1,-0.3) {\small$\theta_{24}^{3}$};
\draw [black, domain=180:210] plot ({0.4*cos(\x)+1.732}, {0.4*sin(\x)}) node at (0.9,-0.21) {\small$\theta_{23}^{4}$};
\draw [black, domain=0:30] plot ({0.4*cos(\x)-1.732}, {0.4*sin(\x)}) node at (-0.9,0.2) {\small $\theta_{14}^{2}$};
\end{tikzpicture} \qquad \qquad
}\quad\quad\quad%5
\subfigure[Infinitesimally weakly rigid framework with constraints of 5 angles.]{ \label{inf_ex06}
\begin{tikzpicture}[scale=0.7]
\node[place] (node1) at (0,1) [label=above:$1$] {};
\node[place] (node2) at (-1.732,0) [label=left:$2$] {};
\node[place] (node3) at (0,-1) [label=below:$3$] {};
\node[place] (node4) at (1.732,0) [label=right:$4$] {};

\draw[dashed] (node1)  -- (node2);
\draw[dashed] (node1)  -- (node4);
\draw[dashed] (node2)  -- (node3);
\draw[dashed] (node3)  -- (node4);
\draw[dashed] (node2)  -- (node4);
\draw [black, domain=210:330] plot ({0.4*cos(\x)}, {0.4*sin(\x) + 1}) node at (0.1,0.3) {\small$\theta_{24}^{1}$};
\draw [black, domain=30:150] plot ({0.4*cos(\x)}, {0.4*sin(\x) - 1}) node at (-0.1,-0.3) {\small$\theta_{24}^{3}$};
\draw [black, domain=180:210] plot ({0.4*cos(\x)+1.732}, {0.4*sin(\x)}) node at (0.9,-0.21) {\small$\theta_{23}^{4}$};
\draw [black, domain=0:30] plot ({0.4*cos(\x)-1.732}, {0.4*sin(\x)}) node at (-0.9,0.2) {\small$\theta_{14}^{2}$};
\draw [black, domain=330:360] plot ({0.5*cos(\x)-1.732}, {0.5*sin(\x)}) node at (-0.9,-0.23) {\small$\theta_{34}^{2}$};
\end{tikzpicture} \qquad
}%6
\caption{Infinitesimally weakly rigid frameworks under different characterizations} \label{inf_ex}
\end{figure*}

For any edge $(i,j) \in \mathcal{E}$ and any angle $(k,i,j) \in \mathcal{A}$, consider the associated relative position vector (edge vector) and cosine defined as
$\mbf{z}_{g} \triangleq \mbf{z}_{ij}, \forall g\in \{ 1,..., m\}$ and
$A_{h} \triangleq \cos{\theta_{h}}, \forall h\in \{ 1,..., q\}$, respectively,
where $\theta_h = \theta_{ij}^{k}$ and $\cos{\theta_{ij}^{k}}
= \left[\frac{\norm{\mbf{z}_{ik}}^{2} + \norm{\mbf{z}_{jk}}^{2} - \norm{\mbf{z}_{ij}}^{2}}{2\norm{\mbf{z}_{ik}}\norm{\mbf{z}_{jk}}}\right]$ induced by the law of cosines.
The \myemph{weak rigidity function} $F_W: \mathbb{R}^{2n} \rightarrow \mathbb{R}^{(m+q)}$ is defined as follows:
$$
F_{W}(\mbf{p}) \triangleq [ \norm{\mbf{z}_{1}}^2, ... ,\norm{\mbf{z}_{m}}^2, A_{1}, ... ,A_{q}]^\top
\in \mathbb{R}^{(m+q)}. 
$$
The weak rigidity function describes the length of edges and subtended angles in the framework.
The \myemph{weak rigidity matrix} is defined as the Jacobian of the weak rigidity function: 
\begin{equation}\label{Weak_rigidity_matrix}
R_{W}(\mbf{p}) \triangleq \frac{\partial F_{W}(\mbf{p})}{\partial \mbf{p}} = \begin{bmatrix}
     \frac{\partial \mathcal{D}}{\partial \mbf{p}}\\ \\
    \frac{\partial \mbf{A}}{\partial \mbf{p}}
  \end{bmatrix} \in 
\mathbb{R}^{(m+q) \times 2n},
\end{equation}
where $\mathcal{D} = [\norm{\mbf{z}_{1}}^2,\norm{\mbf{z}_{2}}^2, ... ,\norm{\mbf{z}_{m}}^2]^\top \in \mathbb{R}^{m}$ and $\mbf{A} = [A_1,A_2,...,A_q]^\top \in \mathbb{R}^{q}$.
Denote $\delta\mbf{p}$ as a variation of the configuration $\mbf{p}$. If $R_W(\mbf{p})\delta\mbf{p}=0$, then $\delta\mbf{p}$ is called an infinitesimal weak motion of $(\mathcal{G},\mbf{p})$. This concept is similar to infinitesimal motions in distance-based rigidity and bearing-based rigidity. Distance preserving motions based on distance rigidity include rigid-body translations and rotations, and bearing preserving motions based on bearing rigidity include rigid-body translations and scalings. On the other hand, the infinitesimal weak motions include not only translations and rotations but also scalings. Figures \ref{inf_ex01} -- \ref{inf_ex05} show that the infinitesimal weak motions include translations and rotations, and Fig. \ref{inf_ex06} shows that the motions include a scaling as well as translations and rotations. 
\begin{definition}[Trivial infinitesimal weak motion] \label{trivial}
An infinitesimal weak motion is called \myemph{trivial} if it corresponds to a translation or a rotation (or a scaling in case of $\mathcal{E} = \emptyset$, for example, see Fig. \ref{inf_ex06}) of the entire framework.
\end{definition}
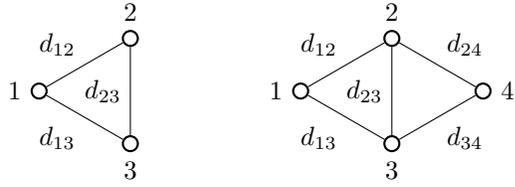
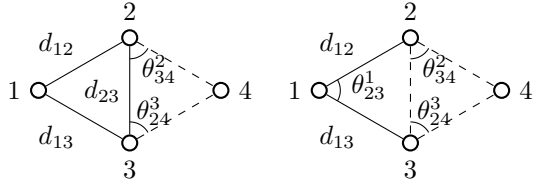
\begin{figure}[h]
\centering
\subfigure[Triangular formation]{ \label{tri_rigid}
\begin{tikzpicture}[scale=0.7]
\node[place] (node1) at (-1.732,0) [label=left:$1$] {};
\node[place] (node2) at (0,1) [label=above:$2$] {};
\node[place] (node3) at (0,-1) [label=below:$3$] {};

\draw[lineUD] (node1)  -- node [above left] {$d_{12}$} (node2);
\draw[lineUD] (node1)  -- node [below left] {$d_{13}$} (node3);
\draw[lineUD] (node2)  -- node [left] {$d_{23}$} (node3);
\end{tikzpicture}\qquad
}\qquad %1
\subfigure[Minimally rigid formation]{ \label{rigid_diamond}
\begin{tikzpicture}[scale=0.7]
\node[place] (node1) at (-1.732,0) [label=left:$1$] {};
\node[place] (node2) at (0,1) [label=above:$2$] {};
\node[place] (node3) at (0,-1) [label=below:$3$] {};
\node[place] (node4) at (1.732,0) [label=right:$4$] {};

\draw[lineUD] (node1)  -- node [above left] {$d_{12}$} (node2);
\draw[lineUD] (node1)  -- node [below left] {$d_{13}$} (node3);
\draw[lineUD] (node2)  -- node [above right] {$d_{24}$} (node4);
\draw[lineUD] (node2)  -- node [left] {$d_{23}$} (node3);
%\draw[dashed] (node2)  -- (node4);
\draw[lineUD] (node3)  -- node [below right] {$d_{34}$} (node4);

\end{tikzpicture}
}%2

\subfigure[Weakly rigid 0-extension]{ \label{0-extension}
\begin{tikzpicture}[scale=0.7]
\node[place] (node1) at (-1.732,0) [label=left:$1$] {};
\node[place] (node2) at (0,1) [label=above:$2$] {};
\node[place] (node3) at (0,-1) [label=below:$3$] {};
\node[place] (node4) at (1.732,0) [label=right:$4$] {};

\draw[lineUD] (node1)  -- node [above left] {$d_{12}$} (node2);
\draw[lineUD] (node1)  -- node [below left] {$d_{13}$} (node3);
\draw[dashed] (node2)  -- (node4);
\draw[lineUD] (node2)  -- node [left] {$d_{23}$} (node3);
%\draw[dashed] (node2)  -- (node4);
\draw[dashed] (node3)  -- (node4);

\draw [black, domain=270:330] plot ({0.4*cos(\x)}, {0.4*sin(\x) + 1}) node at (0.05,0.4) [right] {$\theta_{34}^{2}$};
\draw [black, domain=30:90] plot ({0.4*cos(\x)}, {0.4*sin(\x) - 1}) node at (-0.05,-0.4) [right] {$\theta_{24}^{3}$};
\end{tikzpicture}
}%3
\subfigure[Weakly rigid 1-extension]{ \label{1-extension}
\begin{tikzpicture}[scale=0.7]
\node[place] (node1) at (-1.732,0) [label=left:$1$] {};
\node[place] (node2) at (0,1) [label=above:$2$] {};
\node[place] (node3) at (0,-1) [label=below:$3$] {};
\node[place] (node4) at (1.732,0) [label=right:$4$] {};

\draw[lineUD] (node1)  -- node [above left] {$d_{12}$} (node2);
\draw[dashed] (node2)  -- (node3);
\draw[dashed] (node2)  -- (node4);
\draw[lineUD] (node1)  -- node [below left] {$d_{13}$} (node3);
%\draw[dashed] (node2)  -- (node4);
\draw[dashed] (node3)  -- (node4);

\draw [black, domain=330:390] plot ({0.4*cos(\x) - 1.732}, {0.4*sin(\x)}) node at (-1.3,0.1) [right] {$\theta_{23}^{1}$};
\draw [black, domain=270:330] plot ({0.4*cos(\x)}, {0.4*sin(\x) + 1}) node at (0.05,0.4) [right] {$\theta_{34}^{2}$};
\draw [black, domain=30:90] plot ({0.4*cos(\x)}, {0.4*sin(\x) - 1}) node at (-0.05,-0.4) [right] {$\theta_{24}^{3}$};
\end{tikzpicture}
}%4
\caption{Modified Henneberg constructions.} %\label{inf_ex}
\end{figure}

%%%%%%%%%%%%%%%%%%%%%%%%%%%%%%
\subsection{Infinitesimal Weak Rigidity} 

\begin{definition}[Infinitesimal Weak Rigidity] \label{weak_rigidity_trivial}
A given framework $(\mathcal{G},\mbf{p})$ is \myemph{infinitesimally weakly rigid} in $\mathbb{R}^{2}$ if all the infinitesimal weak motions are trivial.
\end{definition}

Consider a graph $\mathcal{G'} = (\mathcal{V'},\mathcal{E'},\mathcal{A'})$ induced from $\mathcal{G}$ in such a way that:
\begin{itemize}
\item $\mathcal{V'} = \mathcal{V}$,

\item $\mathcal{E'} = \\
\set{(i,j) \given (i,j) \in \mathcal{E} \lor \exists k \in \mathcal{V}\text{ s.t. } (k,i,j) \in \mathcal{A}} \\
\cup \set{(i,k) \given \exists k \in \mathcal{V}\text{ s.t. } (k,i,j) \in \mathcal{A}} \\
  \cup \set{(j,k) \given \exists k \in \mathcal{V}\text{ s.t. } (k,i,j) \in \mathcal{A}}$\\ 
(if $\mathcal{E} = \emptyset$, then $\mathcal{E'} = \emptyset$),

\item $\mathcal{A'} = \mathcal{A}$.
\end{itemize}
For any edge $(i,j) \in \mathcal{E'}$, we consider a new associated relative position vector defined as  
$$
\mbf{z'}_{s} \triangleq \mbf{z'}_{ij}, \forall s\in \{ 1,..., l\}, l \geq m,
$$
where $\mbf{z'}_{ij} = \mbf{p}_{i} - \mbf{p}_{j}$ for all $(i,j)\in \mathcal{E'}$ and $l=\card{\mathcal{E'}}$.
The new associated relative position vector satisfies the following condition:
$$
\mbf{z'}_{u} = \mbf{z}_{u}, \forall u\in \{ 1,..., m\}.
$$ 
Let $\mbf{z}' = \big[\mbf{z'}_{1}^\top, \mbf{z'}_{2}^\top,...,\mbf{z'}_{l}^\top \big]^\top \in \mathbb{R}^{2l}$ denote a new associated column vector composed of relative position vectors.
The oriented incidence matrix $H' \in \mathbb{R}^{l \times n}$ of the new graph $\mathcal{G'}$ is the $\{0, \pm1\}$-matrix with rows indexed by edges and columns indexed by vertices as follows: 
$$
[H']_{ui}=\begin{cases}
    1 & \text{if the $u$-th edge sinks at vertex $i$}\,, \\
    -1 & \text{if the $u$-th edge leaves vertex $i$}\,, \\
	0 & \text{otherwise}\,,
  \end{cases}
$$
where $[H']_{ui}$ is an element at row $u$ and column $i$ of the matrix $H'$. Note that $\mbf{z'}$ satisfies $\mbf{z'}=\bar{H'}\mbf{p}$ where $\bar{H'}\triangleq H'\otimes I_2$.

We first prove a useful expression which will be used later in Lemma \ref{lem_null_of_rigid matrix}. 
\begin{lemma}\label{partial_deriv}
Let $\mbf{z'}_a$, $\mbf{z'}_b$ and $\mbf{z'}_c$ denote relative position vectors to define a cosine $A_{h}$ s.t. $\frac{\norm{\mbf{z'}_{a}}^{2} + \norm{\mbf{z'}_{b}}^{2} - \norm{\mbf{z'}_{c}}^{2}}{2\norm{\mbf{z'}_{a}}\norm{\mbf{z'}_{b}}}$. The following equations hold.
\begin{align}
\frac{\partial A_h}{\partial \mbf{z'}_a}\mbf{z'}_{a} &=\frac{\norm{\mbf{z'}_{a}}^{2} - \norm{\mbf{z'}_{b}}^{2} + \norm{\mbf{z'}_{c}}^{2}}{2\norm{\mbf{z'}_{a}}\norm{\mbf{z'}_{b}}}, \\
\frac{\partial A_h}{\partial \mbf{z'}_b}\mbf{z'}_{b} &=\frac{-\norm{\mbf{z'}_{a}}^{2} + \norm{\mbf{z'}_{b}}^{2} + \norm{\mbf{z'}_{c}}^{2}}{2\norm{\mbf{z'}_{a}}\norm{\mbf{z'}_{b}}}, \\
\frac{\partial A_h}{\partial \mbf{z'}_c}\mbf{z'}_{c} &= -\frac{\norm{\mbf{z'}_{c}}^{2}}{\norm{\mbf{z'}_{a}}\norm{\mbf{z'}_{b}}},
\end{align}
where $a\neq b\neq c$ and $a,b,c \in \{ 1,..., l\}$.
\end{lemma}
\begin{proof}
\begin{figure}[]
\centering
{
\begin{tikzpicture}[scale=0.5]
\node[place] (node1) at (-2,0) [label=left:$i$] {};
\node[place] (node2) at (0,1) [label=above:$k$] {};
\node[place] (node3) at (2,0) [label=right:$j$] {};

\draw[lineUD] (node1)  -- node [above left] {$\mbf{z'}_a$} (node2);
\draw[lineUD] (node1)  -- node [below] {$\mbf{z'}_c$} (node3);
\draw[lineUD] (node2)  -- node [above right] {$\mbf{z'}_b$} (node3);
\end{tikzpicture} \caption{Example of a triangle} \label{Fig:triangle_ijk}
}%1
\end{figure}
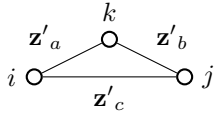 \\
Since $\cos{\theta_{ij}^{k}}
= \left[\frac{\norm{\mbf{z'}_{ik}}^{2} + \norm{\mbf{z'}_{jk}}^{2} - \norm{\mbf{z'}_{ij}}^{2}}{2\norm{\mbf{z'}_{ik}}\norm{\mbf{z'}_{jk}}}\right]$ and $(k,i,j) \in \mathcal{A'}$, with reference to Fig. \ref{Fig:triangle_ijk}, $A_h$ can be expressed as
\begin{align} 
A_{h} = \cos{\theta_{ij}^{k}}
= \frac{\norm{\mbf{z'}_{a}}^{2} + \norm{\mbf{z'}_{b}}^{2} - \norm{\mbf{z'}_{c}}^{2}}{2\norm{\mbf{z'}_{a}}\norm{\mbf{z'}_{b}}}, \nonumber \\ 
\forall a,b,c \in \{ 1,..., l\}, a\neq b\neq c. \nonumber
\end{align}
As a result, the following equations are calculated as
\begin{align} 
\frac{\partial A_h}{\partial \mbf{z'}_a}\mbf{z'}_{a} &= \frac{1}{4\norm{\mbf{z'}_a}^2\norm{\mbf{z'}_b}^2} \bigg[2{\mbf{z'}_a^\top}(2\norm{\mbf{z'}_a}\norm{\mbf{z'}_b})- \nonumber \\ 
&(\norm{\mbf{z'}_a}^2+\norm{\mbf{z'}_b}^2-\norm{\mbf{z'}_c}^2)(2\frac{{\mbf{z'}_a^\top}}{\norm{\mbf{z'}_a}}\norm{\mbf{z'}_b})\bigg]\mbf{z'}_a \nonumber \\ 
&=\frac{\norm{\mbf{z'}_{a}}^{2} - \norm{\mbf{z'}_{b}}^{2} + \norm{\mbf{z'}_{c}}^{2}}{2\norm{\mbf{z'}_{a}}\norm{\mbf{z'}_{b}}},\nonumber \\
\frac{\partial A_h}{\partial \mbf{z'}_c}\mbf{z'}_{c} &= \frac{1}{4\norm{\mbf{z'}_a}^2\norm{\mbf{z'}_b}^2} \big[ -2{\mbf{z'}_c^\top}(2\norm{\mbf{z'}_a}\norm{\mbf{z'}_b})\big]\mbf{z'}_c \nonumber \\ 
&= -\frac{\norm{\mbf{z'}_{c}}^{2}}{\norm{\mbf{z'}_{a}}\norm{\mbf{z'}_{b}}}, \nonumber
\end{align} 
where ${\mbf{z'}_a^\top}\mbf{z'}_a = \norm{\mbf{z'}_a}^2$ and ${\mbf{z'}_c^\top}\mbf{z'}_c = \norm{\mbf{z'}_c}^2$. $\frac{\partial A_h}{\partial \mbf{z'}_b}\mbf{z'}_{b}$ can be also calculated similarly.
\end{proof}

\begin{lemma}\label{Lem:linearly_independence}
If $\mbf{p} \neq 0$, the vectors in the set $L_i=\{\mathds{1}\otimes I_2, (I_n\otimes J)\mbf{p}, \mbf{p} \}$ are linearly independent.
\end{lemma}
\begin{proof}
Let $\{a_1,a_2,a_3,a_4\}$, where $a_1, a_2, a_3, a_4 \in \mathbb{R}^{2n}$, be defined as $\{\mathds{1}\otimes I_2, (I_n\otimes J)\mbf{p}, \mbf{p} \}$. Then, we can set the following equation to determine the linear independence.
\begin{equation}
k_1a_1+k_2a_2+k_3a_3+k_4a_4 = 0, \label{eq:linearly_independence}
\end{equation}
where $k_1, k_2, k_3$ and $k_4$ are scalars. By row-reducing the augmented matrix of equation (\ref{eq:linearly_independence}) and the assumptions that $\mbf{p} \neq 0$ and there are no position vectors collocated at one point, the matrix can be transformed to the reduced row echelon form as follows
$$
\begin{bmatrix}
1 &0 &0 &0 &0 &\cdots &0\\
0 &1 &0 &0 &0 &\cdots &0\\
0 &0 &1 &0 &0 &\cdots &0\\
0 &0 &0 &1 &0 &\cdots &0
\end{bmatrix}^\top.
$$
From the above result, we know that the solution, $k_1= k_2= k_3= k_4= 0$, of equation (\ref{eq:linearly_independence}) is unique. Thus, by the definition of the linearly independence, we can see that the vectors in the set $\{\mathds{1}\otimes I_2, (I_n\otimes J)\mbf{p}, \mbf{p} \}$ are linearly independent.
\end{proof}

\begin{lemma}\label{lem_null_of_rigid matrix}
A framework $(\mathcal{G},\mbf{p})$ in $\mathbb{R}^{2}$ satisfies span$\{\mathds{1}\otimes I_2, (I_n\otimes J)\mbf{p} \} \subseteq$ Null$(R_{W}(\mbf{p}))$ and $\rank(R_{W}(\mbf{p}))\leq 2n-3$ if $\mathcal{E} \neq \emptyset$. If $\mathcal{E} = \emptyset$, then the framework $(\mathcal{G},\mbf{p})$ in $\mathbb{R}^{2}$ satisfies span$\{\mathds{1}\otimes I_2, (I_n\otimes J)\mbf{p}, \mbf{p} \}\subseteq$ Null$(R_{W}(\mbf{p}))$ and $\rank(R_{W}(\mbf{p}))\leq 2n-4$.
\end{lemma}

\begin{proof}
%The rigidity matrix (equation (\ref{rigidity_D_matrix})) can be expressed as\cite{sun2015rigid}
%$$
%R_D(\mbf{p})=\frac{1}{2} \frac{\partial F_D(\mbf{p})}{\partial \mbf{p}}= \mbf{D}(\mbf{z})^\top(H\otimes I_2)
%$$
If $\mathcal{E} \neq \emptyset$ (for example, from Fig. \ref{inf_ex01} to Fig. \ref{inf_ex05}), then the equation (\ref{Weak_rigidity_matrix}) can be expressed as follows
$$
R_{W}(\mbf{p}) = \frac{\partial F_{W}(\mbf{p})}{\partial \mbf{p}}=  \begin{bmatrix}
     \frac{\partial \mathcal{D}}{\partial \mbf{z'}} \frac{\partial \mbf{z'}}{\partial \mbf{p}}\\ \\
    \frac{\partial \mbf{A}}{\partial \mbf{z'}} \frac{\partial \mbf{z'}}{\partial \mbf{p}}
  \end{bmatrix} = \begin{bmatrix}
     \frac{\partial \mathcal{D}}{\partial \mbf{z'}} \bar{H'} \\ \\
    \frac{\partial \mbf{A}}{\partial \mbf{z'}} \bar{H'}
  \end{bmatrix} = \begin{bmatrix}
    \frac{\partial \mathcal{D}}{\partial \mbf{z'}}\\ \\
    \frac{\partial \mbf{A}}{\partial \mbf{z'}}
  \end{bmatrix}  \bar{H'},
$$
where $\mbf{A} = [A_1,A_2,...,A_q]^\top \in \mathbb{R}^{q}$. First, it is clear that span$\{\mathds{1}\otimes I_2\}$ $\subseteq$ Null$(\bar{H'})$ $\subseteq$ Null$(R_{W}(\mbf{p}))$. 
Second, $\bar{H'}(I_n\otimes J)\mbf{p}$ can be expressed as
\begin{align}
\bar{H'}(I_n\otimes J)\mbf{p} &= (H'\otimes I_2)(I_n\otimes J)\mbf{p} =(H'\otimes J)\mbf{p} \nonumber \\
&= (I_l H')\otimes(J I_2)\mbf{p} = (I_l \otimes J)(H'\otimes I_2)\mbf{p} \nonumber \\
&=(I_l \otimes J)\mbf{z'}= \begin{bmatrix}
J\mbf{z'}_1 \\
\vdots\\
J\mbf{z'}_l
\end{bmatrix}. \nonumber
\end{align}
Let $A_h$ be an element of vector $\mbf{A}$ for $h\in \{ 1,..., q\}$ as mentioned in Lemma \ref{partial_deriv}. Then, the elements of $\frac{\partial A_h}{\partial \mbf{z'}}$ are zero except for $\frac{\partial A_h}{\partial \mbf{z'}_a}$, $\frac{\partial A_h}{\partial \mbf{z'}_b}$ and $\frac{\partial A_h}{\partial \mbf{z'}_c}$, where $a,b,c \in \{ 1,..., l\}$. With reference to the calculation of $\frac{\partial A_h}{\partial \mbf{z'}_a}$ in Lemma \ref{partial_deriv},
$\frac{\partial A_h}{\partial \mbf{z'}}\bar{H'}(I_n\otimes J)\mbf{p}$ is calculated as
\begin{align}
\frac{\partial A_h}{\partial \mbf{z'}}\bar{H'}(I_n\otimes J)\mbf{p}=\frac{\partial A_h}{\partial \mbf{z'}}(I_l \otimes J)\mbf{z'} = \frac{\partial A_h}{\partial \mbf{z'}}\begin{bmatrix}
J\mbf{z'}_1 \\
\vdots\\
J\mbf{z'}_l
\end{bmatrix} \nonumber \\
=\frac{\partial A_h}{\partial \mbf{z'}_a}J\mbf{z'}_{a}+\frac{\partial A_h}{\partial \mbf{z'}_b}J\mbf{z'}_{b}+\frac{\partial A_h}{\partial \mbf{z'}_c}J\mbf{z'}_{c} = 0, \nonumber
\end{align}
where ${\mbf{z'}_a}^\top J \mbf{z'}_a = 0$, ${\mbf{z'}_b}^\top J \mbf{z'}_b = 0$ and ${\mbf{z'}_c}^\top J \mbf{z'}_c = 0$. Thus, $\frac{\partial \mbf{A}}{\partial \mbf{z'}}\bar{H'}(I_n\otimes J)\mbf{p} = 0$. Also, the following equation is calculated as
\begin{align}
\frac{\partial \mathcal{D}}{\partial \mbf{z'}}\bar{H'}(I_n\otimes J)\mbf{p} 
&= \frac{\partial \mathcal{D}}{\partial \mbf{z'}}(I_l \otimes J)\mbf{z'} 
= \frac{\partial \mathcal{D}}{\partial \mbf{z'}}\begin{bmatrix}
J\mbf{z'}_1 \\
\vdots\\
J\mbf{z'}_l
\end{bmatrix} \nonumber \\
&= \begin{bmatrix} 
2D^\top & 0_{m,(2l-2m)}
\end{bmatrix} \begin{bmatrix}
J\mbf{z'}_1 \\
\vdots\\
J\mbf{z'}_l
\end{bmatrix} =0, \nonumber
\end{align} 
where $D=$diag$(\mbf{z'}_1,...,\mbf{z'}_m) \in \mathbb{R}^{2m \times m}$ and $0_{m,(2l-2m)}$ is a $m \times (2l-2m)$ zero matrix.
%\begin{bmatrix} 
%2D^\top & 0 & \cdots & 0 \\ 
%0 & 0 & \cdots & 0 \\
%\vdots & \vdots & \ddots & \vdots \\
%0 & 0 & \cdots & 0
%\end{bmatrix}
Using the above results, the following equation can be calculated as
\begin{align}
R_{W}(\mbf{p})(I_n\otimes J)\mbf{p} &= \begin{bmatrix}
    \frac{\partial \mathcal{D}}{\partial \mbf{z'}}\\ \\
    \frac{\partial \mbf{A}}{\partial \mbf{z'}}
  \end{bmatrix} \bar{H'} (I_n\otimes J)\mbf{p} \nonumber \\
&=\begin{bmatrix}
    \frac{\partial \mathcal{D}}{\partial \mbf{z'}} \\ \\
    \frac{\partial \mbf{A}}{\partial \mbf{z'}}
  \end{bmatrix} \begin{bmatrix}
J\mbf{z'}_1 \\
\vdots\\
J\mbf{z'}_l
\end{bmatrix} = 0. \nonumber 
\end{align}
Therefore, we have span$\{(I_n\otimes J)\mbf{p} \} \subseteq$ Null$(R_{W}(\mbf{p}))$. Also, with span$\{\mathds{1}\otimes I_2\}\subseteq$ Null$(R_{W}(\mbf{p}))$ and Lemma \ref{Lem:linearly_independence}, span$\{\mathds{1}\otimes I_2, (I_n\otimes J)\mbf{p} \} \subseteq$ Null$(R_{W}(\mbf{p}))$ holds and the inequality $\rank(R_{W}(\mbf{p}))\leq 2n-3$ is expressed from span$\{\mathds{1}\otimes I_2, (I_n\otimes J)\mbf{p} \} \subseteq$ Null$(R_{W}(\mbf{p}))$ directly. 

However, if $\mathcal{E} = \emptyset$ (for example, Fig. \ref{inf_ex06}), then the equation (\ref{Weak_rigidity_matrix}) can be expressed as
$$
R_{W}(\mbf{p}) = \frac{\partial F_{W}(\mbf{p})}{\partial \mbf{p}}=
    \frac{\partial \mbf{A}}{\partial \mbf{z'}} \bar{H'}.
$$
Then, $R_{W}(\mbf{p})\mbf{p} = \frac{\partial \mbf{A}}{\partial \mbf{z'}} \bar{H'}\mbf{p} = \frac{\partial \mbf{A}}{\partial \mbf{z'}}\mbf{z'}$. 
The elements of $\frac{\partial A_h}{\partial \mbf{z'}}$ are zero except for $\frac{\partial A_h}{\partial \mbf{z'}_a}$, $\frac{\partial A_h}{\partial \mbf{z'}_b}$ and $\frac{\partial A_h}{\partial \mbf{z'}_c}$. With Lemma \ref{partial_deriv},
$\frac{\partial A_h}{\partial \mbf{z'}}\mbf{z'}$ is calculated as follows
\begin{align}
\frac{\partial A_h}{\partial \mbf{z'}}\mbf{z'} &= \frac{\partial A_h}{\partial \mbf{z'}}\begin{bmatrix}
\mbf{z'}_1 \\
\vdots\\
\mbf{z'}_l
\end{bmatrix}
=
\frac{\partial A_h}{\partial \mbf{z'}_a}\mbf{z'}_{a}+\frac{\partial A_h}{\partial \mbf{z'}_b}\mbf{z'}_{b}+\frac{\partial A_h}{\partial \mbf{z'}_c}\mbf{z'}_{c} \nonumber \\
&= \frac{\norm{\mbf{z'}_{a}}^{2} - \norm{\mbf{z'}_{b}}^{2} + \norm{\mbf{z'}_{c}}^{2}}{2\norm{\mbf{z'}_{a}}\norm{\mbf{z'}_{b}}} + \nonumber \\
&\frac{-\norm{\mbf{z'}_{a}}^{2} + \norm{\mbf{z'}_{b}}^{2} + \norm{\mbf{z'}_{c}}^{2}}{2\norm{\mbf{z'}_{a}}\norm{\mbf{z'}_{b}}} +
\frac{-2\norm{\mbf{z'}_{c}}^{2}}{2\norm{\mbf{z'}_{a}}\norm{\mbf{z'}_{b}}} = 0. \nonumber
\end{align}
Therefore, $R_{W}(\mbf{p})\mbf{p}=\frac{\partial \mbf{A}}{\partial \mbf{z'}}\mbf{z'}=0$ and $\mbf{p}\subseteq$ Null$(R_{W}(\mbf{p}))$. Also, we can easily prove that span$\{\mathds{1}\otimes I_2, (I_n\otimes J)\mbf{p} \} \subseteq$ Null$(R_{W}(\mbf{p}))$ as the case of $\mathcal{E} \neq \emptyset$ is proved. With Lemma \ref{Lem:linearly_independence}, we can see that the inequality $\rank(R_{W}(\mbf{p}))\leq 2n-4$ is expressed from span$\{\mathds{1}\otimes I_2, (I_n\otimes J)\mbf{p}, \mbf{p} \} \subseteq$ Null$(R_{W}(\mbf{p}))$ directly.
\end{proof}

The next result gives us the necessary and sufficient condition for infinitesimal weak rigidity of a framework.
\begin{theorem}
\label{Thm:Inf_Rank1}
A framework $(\mathcal{G},\mbf{p})$ with $n \ge 3$ and $\mathcal{E} \neq \emptyset$ is infinitesimally weakly rigid in $\mathbb{R}^{2}$ if and only if the weak rigidity matrix $R_{W}(\mbf{p})$ has rank $2n - 3$.
\end{theorem}
\begin{proof}
%All weakly infinitesimal motions must be in the null space of weak rigidity matrix since the weak rigidity matrix expresses of all constraints on the weakly infinitesimal motions. The following equation must be satisfied, rank$(R_W)=2\card{\mathcal{V}}-$nullity$(R_W)$. If Null$(R_W)$ are weakly infinitesimal motions and the motions are trivial (a translation and a rotation of entire framework), then the rank condition of the weak rigidity matrix satisfies $2\card{\mathcal{V}}-3$. Otherwise (If the null space of $R_W$ contains any nontrivial weakly infinitesimal motions), the rank must be less than $2\card{\mathcal{V}}-3$.
Lemma \ref{lem_null_of_rigid matrix} shows span$\{\mathds{1}\otimes I_2, (I_n\otimes J)\mbf{p} \} \subseteq$ Null$(R_{W}(\mbf{p}))$. Observe that $\mathds{1}\otimes I_2$ and $(I_n\otimes J)\mbf{p}$ correspond to a rigid-body translation and a rotation of the framework, respectively, with reference to \cite{zhao2016bearing,sun2017distributed}. Therefore, the theorem directly follows from Definition \ref{weak_rigidity_trivial}.
\end{proof}

However, in case of $\mathcal{E} = \emptyset$, the condition span$\{\mathds{1}\otimes I_2, (I_n\otimes J)\mbf{p}, \mbf{p} \}\subseteq$ Null$(R_{W}(\mbf{p}))$ is satisfied as proved in Lemma \ref{lem_null_of_rigid matrix}. Observe that $\mathds{1}\otimes I_2$, $(I_n\otimes J)\mbf{p}$ and $\mbf{p}$ correspond to a rigid-body translation, a rotation and a scaling of the framework, respectively, with reference to \cite{zhao2016bearing,sun2017distributed}. Therefore, when $\mathcal{E} = \emptyset$, the following theorem follows from Definition \ref{weak_rigidity_trivial} directly.

\begin{theorem}
\label{Thm:Inf_Rank2}
A framework $(\mathcal{G},\mbf{p})$ with $n \ge 3$ and $\mathcal{E} = \emptyset$ is infinitesimally weakly rigid in $\mathbb{R}^{2}$ if and only if the weak rigidity matrix $R_{W}(\mbf{p})$ has rank $2n - 4$.
\end{theorem}

%%%%%%%%%%%%%%%%%%%%%%%%%%%%%%%%%%%%%%%%%%%%%%%%%%%%%%%%%%%%%%%%%%%%%%%%%%%%%%%%
\section{The Formation Control Problem on Three-Agent Formations}
\label{Sec:Formation Control Problem}
%%%%%%%%%%%%%%%%%%%%%%%%%%%%%%

Let $t \in [0,\infty)$ be time. We assume that the motion of an agent $i$ is governed by a single integrator, i.e.,
\begin{equation}
\frac{d}{dt}\mbf{p}_i=\dot{\mbf{p}}_i=u_i,
\end{equation}
where $u_i$ is a control input. Define the following two column vectors of squared distances and cosines
\begin{align}
d_c(\mbf{p}) &= [\ldots, d_{ij}^2, \ldots]^\top_{(i,j) \in \mathcal{E}}, \\
c_c(\mbf{p}) &= [\ldots, \cos\theta_{ij}^k, \ldots]^\top_{(k,i,j) \in \mathcal{A}}.
\end{align}
Similarly, $d_c^*$ and $c_c^*$ are defined as vectors of desired squared distance constraints and cosine constraints respectively, and both of them are constant. Then, an error vector can be defined as follows
\begin{equation}
\mbf{e}(\mbf{p})=[d_c(\mbf{p})^\top c_c(\mbf{p})^\top]^\top - [d_c^{*\top} c_c^{*\top}]^\top.
\end{equation}
If either $\mathcal{E}=\emptyset$ or $\mathcal{A}=\emptyset$, then the error vector is $\mbf{e}(\mbf{p})=d_c(\mbf{p})-d_c^*$ or $\mbf{e}(\mbf{p})=c_c(\mbf{p})-c_c^*$ respectively. 

We consider the following formation control problem.
\begin{problem}
\label{problem1}
The weakly rigid formation control problem is to design a control input $u_i$, $\forall i \in \mathcal{V}$, such that $\mathbf{e} \to 0$ as t $\to \infty$.
\end{problem}
Since we only consider a three-agent formation problem with two distance constraints and one angle constraint. The error vector is written as $\mbf{e}(\mbf{p})=[e_{12} \quad e_{13} \quad e^1_{23}]^\top$, where $e_{ij}=d_{ij}^2-d_{ij}^{*2}$ and $e^k_{ij}=\cos\theta^k_{ij}-\cos(\theta^k_{ij})^*$.
%%%%%%%%%%%%%%%%%%%%%%%%%%%%%%
\subsection{Equations of motion}
The gradient-descent law \cite{krick2009stabilisation,bishop2015distributed,park2014stability} is employed to make a formation control system stable. First, we consider the control law defined as 
\begin{equation}
\dot{\mbf{p}}=u \triangleq -(\nabla\mbf{e}(\mbf{p}))^\top\mbf{e}(\mbf{p}).
\end{equation}
The control law can be written as
\begin{subequations}
\begin{align}
\dot{\mbf{p}}=u = -(\nabla\mbf{e})^\top\mbf{e} 
&= -\frac{\partial}{\partial \mbf{p}}\left(\begin{bmatrix}
d_c(\mbf{p})\\
c_c(\mbf{p})
\end{bmatrix} - \begin{bmatrix}
d_c^*\\
c_c^*
\end{bmatrix}\right)^\top \mbf{e}(\mbf{p}), \nonumber
 \\
&= -\frac{\partial}{\partial \mbf{p}}\left(\begin{bmatrix}
d_c(\mbf{p})\\
c_c(\mbf{p})
\end{bmatrix}\right)^\top \mbf{e}(\mbf{p}),
\\
&= -R_{W}(\mbf{p})^\top \mbf{e}(\mbf{p}) \label{control_law01}.
\end{align} \label{controller}
\end{subequations} 
In the case of the three-agent formation, the equation (\ref{control_law01}) can be again written as
\begin{align}
\dot{\mbf{p}}=u
&= -R_{W}(\mbf{p})^\top \mbf{e}(\mbf{p}) \nonumber
\\
&= -(E(\mbf{p})\otimes I_2)\mbf{p}
\end{align}
where $R_{W}(\mbf{p})$ is defined by
$$
R_{W}(\mbf{p})=\begin{bmatrix}
2\mbf{z}_{12}^\top & -2\mbf{z}_{12}^\top 	&		0			\\
2\mbf{z}_{13}^\top & 		0 			&-2\mbf{z}_{13}^\top 	\\
\alpha				& \beta			& \gamma
\end{bmatrix},
$$
$$
\alpha=\frac{\partial}{\partial \mbf{p}_1}\cos\theta^1_{23}, \;
\beta=\frac{\partial}{\partial \mbf{p}_2}\cos\theta^1_{23}, \;
\gamma=\frac{\partial}{\partial \mbf{p}_3}\cos\theta^1_{23},
$$ $\alpha,\beta,\gamma \in \mathbb{R}^{1 \times 2}$,
and $E(\mbf{p})=$
$$
\small	%%font size = small size
\begin{bmatrix}
2e_{12}+2e_{13}+\alpha_{\mbf{p}_1}e_{23}^1 & -2e_{12}+\alpha_{\mbf{p}_2}e_{23}^1 & -2e_{13}+\alpha_{\mbf{p}_3}e_{23}^1\\

-2e_{12}+\beta_{\mbf{p}_1}e_{23}^1 & 2e_{12}+\beta_{\mbf{p}_2}e_{23}^1 & \beta_{\mbf{p}_3}e_{23}^1 \\

-2e_{13}+\gamma_{\mbf{p}_1}e_{23}^1 & \gamma_{\mbf{p}_2}e_{23}^1 & 2e_{13}+\gamma_{\mbf{p}_3}e_{23}^1
\end{bmatrix}.
$$
\normalsize %%font size = normal size
In the matrix $E(\mbf{p})$, $\alpha_{\mbf{p}_1}$, $\alpha_{\mbf{p}_2}$ and $\alpha_{\mbf{p}_3}$ are coefficients of $\mbf{p}_1$, $\mbf{p}_2$ and $\mbf{p}_3$ in $\alpha$, respectively. Similarly, $\beta_{\mbf{p}_1}$, $\beta_{\mbf{p}_2}$, $\beta_{\mbf{p}_3}$, $\gamma_{\mbf{p}_1}$, $\gamma_{\mbf{p}_2}$ and $\gamma_{\mbf{p}_3}$ are defined. Also, equations of $\alpha_{\mbf{p}_2}=\beta_{\mbf{p}_1}$, $\alpha_{\mbf{p}_3}=\gamma_{\mbf{p}_1}$ and $\beta_{\mbf{p}_3}=\gamma_{\mbf{p}_2}$ hold, i.e. the matrix $E(\mbf{p})$ is symmetric.

We define a desired equilibrium set and an incorrect equilibrium set as
\begin{align}
\mathcal{P}^* &= \set{\mbf{p} \in \mathbb{R}^{2n} \given \mbf{e}=0}, \\
\mathcal{P}_i &= \set{\mbf{p} \in \mathbb{R}^{2n} \given R_{W}^\top\mbf{e}=0, \mbf{e}\neq0},
\end{align}
respectively. The first set $\mathcal{P}^*$ corresponds to a desired target formation, and the second set $\mathcal{P}_i$ does not correspond to a desired target formation but makes the equation (\ref{control_law01}) become zero. Both of the sets constitute the set of all equilibria.
%%%%%%%%%%%%%%%%%%%%%%%%%%%%%%
\subsection{Analysis of the incorrect equilibrium points}
%Let configuration $\mbf{p}_*$ be in the incorrect equilibrium set $\mathcal{P}_i$. Then, $R_{W}(\mbf{p}_*)^\top\mbf{e}(\mbf{p}_*)=0$ but $\mbf{e}(\mbf{p}_*) \neq 0$. It means that $R_{W}(\mbf{p}_*)$ does not have rank $2n-3$ when $\mathcal{E} \neq \emptyset$ (or $2n-4$ when $\mathcal{E} = \emptyset$), i.e.,  the framework $(\mathcal{G},\mbf{p}_*)$ is not infinitesimally weakly rigid. Those incorrect equilibria occur when the 3 agents are collinear as the following Lemma.
\begin{lemma}\label{incorrect_collinear}
In the case of the three-agent formation, incorrect equilibria take place only when the three agents are collinear.
\end{lemma}
\begin{proof}
The equation (\ref{control_law01}) can be written as
\begin{subequations}
\begin{align}
\dot{\mbf{p}}_1 &= -2\mbf{z}_{12}e_{12}-2\mbf{z}_{13}e_{13}-\alpha^\top e^1_{23} \label{collinear_01} \\ 
\dot{\mbf{p}}_2 &= 2\mbf{z}_{12}e_{12}-\beta^\top e^1_{23} \label{collinear_02} \\ 
\dot{\mbf{p}}_3 &= 2\mbf{z}_{13}e_{13}-\gamma^\top e^1_{23} \label{collinear_03}
\end{align}
\end{subequations}
In the incorrect equilibrium set $\mathcal{P}_i$, the equation (\ref{collinear_03}) is calculated as
\begin{equation}
\mbf{z}_{12}=\left(\frac{\norm{\mbf{z}_{12}}}{\norm{\mbf{z}_{13}}}\cos\theta^1_{23}-2\norm{\mbf{z}_{12}}\norm{\mbf{z}_{13}}\frac{e_{13}}{e^1_{23}} \right) \mbf{z}_{13} \mid_{\mbf{p} \in \mathcal{P}_i} \label{Eq:collinear_04}
\end{equation}
It follows that $\mbf{p}_1$, $\mbf{p}_2$ and $\mbf{p}_3$ must be collinear from the equation (\ref{Eq:collinear_04}).
The equations (\ref{collinear_01}) and (\ref{collinear_02}) also give us similar results. We assumed that there are no position vectors overlapping each other. Because the cosine cannot be defined if there exists at least one overlapped point of two agents. Thus, $e^1_{23}$ cannot be equal to zero in the incorrect equilibrium set $\mathcal{P}_i$ and, regardless of the values of $e_{12}$ and $e_{13}$, the three agents must be collinear.
The formation shape of the three agents falls into one of three cases as depicted in Fig. \ref{Fig:formation_forms}.
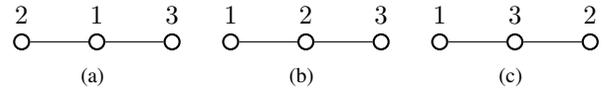
\begin{figure}[h]
\centering
\subfigure[]{ %\label{a}
\begin{tikzpicture}[scale=1]
\node[place] (node2) at (-1,0) [label=above:$2$] {};
\node[place] (node1) at (0,0) [label=above:$1$] {};
\node[place] (node3) at (1,0) [label=above:$3$] {};

\draw[lineUD] (node1)  -- node {}(node2);
\draw[lineUD] (node1)  -- node {}(node3);
\end{tikzpicture}
} %1
\subfigure[]{ %\label{b}
\begin{tikzpicture}[scale=1]
\node[place] (node1) at (-1,0) [label=above:$1$] {};
\node[place] (node2) at (0,0) [label=above:$2$] {};
\node[place] (node3) at (1,0) [label=above:$3$] {};

\draw[lineUD] (node1)  -- node {}(node2);
\draw[lineUD] (node2)  -- node {}(node3);
\end{tikzpicture}
} %2
\subfigure[]{ %\label{c}
\begin{tikzpicture}[scale=1]
\node[place] (node1) at (-1,0) [label=above:$1$] {};
\node[place] (node3) at (0,0) [label=above:$3$] {};
\node[place] (node2) at (1,0) [label=above:$2$] {};

\draw[lineUD] (node1)  -- node {}(node3);
\draw[lineUD] (node2)  -- node {}(node3);
\end{tikzpicture}
} %3
\caption{Three formation forms which can occur at the incorrect equilibria.} \label{Fig:formation_forms}
\end{figure}
\end{proof}
Note that the stability of an equilibrium point is independent of a rigid-body translation, a rotation and a scaling of a framework. Because relative distances and subtended angles only matter. Therefore, without loss of generality, we suppose that the three agents are on the x-axis to analyze the stability at the incorrect equilibria. Also, it is observed that $\frac{\partial}{\partial p_i} \cos\theta_{23}^1 = -\sin\theta_{23}^1 \frac{\partial \theta_{23}^1}{\partial p_i}$. Thus, if the three agents are in a collinear configuration, there holds $\theta_{23}^1 = 0$ or $\pi$, which implies that the values of $\alpha$, $\beta$, and $\gamma$ calculated at an incorrect equilibrium are 0.

To analyze the stability at the incorrect equilibria, we linearize the system (\ref{control_law01}). The negative Jacobian $J(\mbf{p})$ of the system (\ref{control_law01}) with respect to $\mbf{p}$ is given by
\begin{align}
J(\mbf{p})=-\frac{\partial}{\partial \mbf{p}}\dot{\mbf{p}}=R_{W}(\mbf{p})^\top R_{W}(\mbf{p})+E(\mbf{p})\otimes I_2 \nonumber \\ 
+ \sum_{i=1}^3(I_3\otimes\mbf{p}_i) \frac{\partial}{\partial \mbf{p}} \begin{bmatrix}
\alpha_{\mbf{p}_i}\\
\beta_{\mbf{p}_i}\\
\gamma_{\mbf{p}_i}
\end{bmatrix} e^1_{23}. \label{negative_jaco}
\end{align}
If $J(\mbf{p})$ has a negative eigenvalue at the incorrect equilibrium point, then the system at the incorrect equilibrium is unstable. We also use a permutation matrix $T$ which reorders columns of matrix such that
\begin{align}
R_WT&=[R_x \,\,\, R_y]=\bar{R}_W, \nonumber \\ 
P_1T&=[P_{1x} \,\,\, P_{1y}]=\bar{P}_1, \nonumber \\ 
C_1T&=[C_{1x} \,\,\, C_{1y}]=\bar{C}_1, \nonumber 
\end{align}
where $R_i \in \mathbb{R}^{3\times3}$, $P_{1i} \in \mathbb{R}^{3\times3}$ and $C_{1i} \in \mathbb{R}^{3\times3}$ are matrices whose columns are composed of the columns of coordinate $i$ in the matrix $R_W$, $P_1$ and $C_1$ respectively, and $P_1\triangleq(I_3\otimes \mbf{p}_1^\top)$, $C_1 \triangleq \frac{\partial}{\partial \mbf{p}}[\alpha_{\mbf{p}_1} \,\, \beta_{\mbf{p}_1} \,\, \gamma_{\mbf{p}_1}]^\top$. Similarly, $\bar{P}_2$, $\bar{P}_3$, $\bar{C}_2$ and $\bar{C}_3$ are defined in the same way.
\begin{lemma}\label{incorrect_negative_eigen}
Let $\mbf{p}_*$ be in the incorrect equilibrium set $\mathcal{P}_i$. Then, $E(\mbf{p}_*)$ has at least one negative eigenvalue.
\end{lemma}
\begin{proof}
Consider a configuration $\bar{\mbf{p}}$, then the following equation holds
$$
\bar{\mbf{p}}^\top[E(\mbf{p}_*)\otimes I_3]\bar{\mbf{p}} = \sum_{(i,j)\in\mathcal{E}}\mbf{e}_{ij}(\mbf{p}_*)\norm{\bar{\mbf{p}}_i-\bar{\mbf{p}}_j}^2,
$$
where the parts involving $\mbf{e}^1_{23}$ vanished since $\alpha|_{p=p_*}=0$, $\beta|_{p=p_*}=0$, and $\gamma|_{p=p_*}=0$. The remaining of this proof is similar to the proof of Lemma 1 in \cite{park2014stability}.
\end{proof}

The permutated matrix $\bar{J}(\mbf{p})$ is given by
\begin{align}
\bar{J}(\mbf{p}) &= T^\top J(\mbf{p}) T   \nonumber \\
&= \bar{R}_W^\top \bar{R}_W + I_2\otimes E(\mbf{p})+ \sum_{i=1}^3(\bar{P}_i^\top \bar{C}_i)e^1_{23} \nonumber \\
&= \begin{bmatrix}
\bar{J}_{11} & \bar{J}_{12} \\
\bar{J}_{21} & \bar{J}_{22}
\end{bmatrix}, \nonumber
\end{align}
where \\ \small
$\bar{J}_{11} = R_x^\top R_x+E(\mbf{p})+\sum_{i=1}^3P_{ix}C_{ix}e^1_{23}$, \\ 
$\bar{J}_{12} = R_x^\top R_y+\sum_{i=1}^3P_{ix}C_{iy}e^1_{23}$, \\
$\bar{J}_{21} = R_y^\top R_x+\sum_{i=1}^3P_{iy}C_{ix}e^1_{23}$, \\
$\bar{J}_{22} = R_y^\top R_y+E(\mbf{p})+\sum_{i=1}^3P_{iy}C_{iy}e^1_{23}$. \\
\normalsize
\begin{theorem}
\label{Thm:incorrect_unstable}
The system (\ref{controller}) at any incorrect equilibrium point $\mbf{p}_*$ is unstable.
\end{theorem}
\begin{proof}
From Lemma \ref{incorrect_collinear}, three agents in the incorrect equilibrium set $\mathcal{P}_i$ are collinear. The stability is also independent on a rigid-body translation, a rotation of the formation. Therefore, assuming that the formation is on the x-axis, the permutated matrix $\bar{J}(\mbf{p}_*)$ is given by
$$
\bar{J}(\mbf{p}_*) = \begin{bmatrix}
R_x^\top R_x+E(\mbf{p_*})+\sum_{i=1}^3P_{ix}C_{ix}e^1_{23} & 0 \\
0 & E(\mbf{p_*})
\end{bmatrix}.
$$
From Lemma \ref{incorrect_negative_eigen}, we know that $E(\mbf{p}_*)$ has at least one negative eigenvalue and the matrix $\bar{J}(\mbf{p}_*)$ also does. Since eigenvalues of $\bar{J}(\mbf{p_*})$ and $J(\mbf{p_*})$ are the same, $J(\mbf{p_*})$ also has at least one negative eigenvalue. Thus, the system (\ref{controller}) at any incorrect equilibrium point $\mbf{p}_*$ is unstable
\end{proof}

\begin{lemma}\label{Lem:no_apporach}
Let $\mbf{p}(0)$ denote an initial position, and $Z$ and $\mathcal{C}$ are defined as 
$Z			=	[\mbf{z}_{12} \quad  \mbf{z}_{13}] \in \mathbb{R}^{2 \times 2}$,
$\mathcal{C}	=	\set{\mbf{p} \in \mathbb{R}^{2n} \given \det Z = 0}$, respectively.
If $\mbf{p}(0)$ is not in $\mathcal{C}$, then $\mbf{p}(t)$ does not approach $\mathcal{P}_i$ for any time $t \geq t_0$.
\end{lemma}
\begin{proof}
First, $\mbf{z}_{12} - \mbf{z}_{13} + \mbf{z}_{23} = 0$ and it follows that $\det[\mbf{z}_{12} \quad \mbf{z}_{13}]=\det[\mbf{z}_{12} \quad  \mbf{z}_{23}]$, which implies that $\mathcal{P}_i \subset \mathcal{C}$ from \cite{cao2007controlling}. We have the following derivative:
$\frac{d}{dt} \det Z =\frac{d}{dt}\det[\mbf{z}_{12} \quad \mbf{z}_{13}]=- \sigma \det Z$, where $\sigma(t)=4\mbf{e}_{12}+4\mbf{e}_{13}-\cos\theta_{23}^1(\frac{1}{\norm{z_{12}}^2}+\frac{1}{\norm{z_{13}}^2})\mbf{e}_{23}^1$. Thus, if $\det Z({t_0}) = 0$ then $\det Z(t) = 0, \forall t \geq 0$. The remainder of this proof is similar to the proof of Lemma 5 in \cite{park2013control}.
\end{proof}

\begin{theorem}[Stability]
\label{Thm:desired_eq_set_stable}
If $\mbf{p}(0)$ is not in $\mathcal{C}$ and $\mathcal{P}_i$, then the $\mbf{p}$ exponentially converges to a point in the desired equilibrium set $\mathcal{P}^*$.
\end{theorem}
\begin{proof}
We define a Lyapunov candidate function as
$
V(\mbf{e})=\frac{1}{2}\mbf{e}^\top\mbf{e}$. Notice that $V(\mbf{e}) \geq 0 \,\, \text{with} \,\, V(\mbf{e})=0 \,\, \text{iff} \,\, \mbf{e} = 0, 
$ and $V$ is radially unbounded. The error dynamics can be written by
\begin{align}
\dot{\mbf{e}} &= R_W(\mbf{e})\dot{\mbf{e}} = -R_W(\mbf{e}) R_W(\mbf{e})^\top \mbf{e}. \nonumber
\end{align}
Then, the derivative of $V(\mbf{e})$ along a trajectory of $\mathbf{e}$ is calculated as 
\begin{align}
\dot{V} = \mbf{e}^\top \dot{\mbf{e}} 
=-\mbf{e}^\top R_W R_W^\top \mbf{e}
= - \norm{R_W^\top \mbf{e}}^2. \label{eq:dot_potential_fn}
\end{align} 
We know that $\dot{V} \leq 0$, $\dot{V}$ is equal to zero iff $R_{W}^\top\mbf{e}=0$. From Theorem \ref{Thm:incorrect_unstable} and Lemma \ref{Lem:no_apporach},  and the assumption that $\mbf{p}(0) \notin \mathcal{P}_i$, it follows that $\mbf{e} \to 0$ asymptotically fast.

Moreover, it follows from $\mbf{p}(0) \notin \mathcal{P}_i$ that the initial positions are not collinear. Thus, the formation is weakly rigid and the rigidity matrix associated with a graph ${K}_3$ has full row rank from Corollary 1 of \cite{park2017rigidity}. It follows that the formation has only two distance preserving motions, i.e. a translation and a rotation. Also, we intuitively know that the infinitesimal weak motions in the case of $\mathcal{E}\neq 0$ correspond to the distance preserving motions with respect to the same formation. In this regard, the two matrices have the same null space, and thus $R_{W}(\mbf{p})$ also has full row rank for all $\mbf{p} \notin \mathcal{P}_i$. It follows from $\mbf{p}(0) \notin \mathcal{P}_i$ and Lemma \ref{Lem:no_apporach} that $R_W R_W^\top$ is positive definite, $\forall t \geq 0$. Henceforth, along a trajectory of $\mbf{e}$, the equation (\ref{eq:dot_potential_fn}) satisfies
\begin{align}
\dot{V} \leq  -\lambda_{\text{min}}(R_W R_W^\top) \norm{\mbf{e}}^2, \nonumber
\end{align}
where $\lambda_{\text{min}}$ denotes the minimum eigenvalue of $R_w R_w^\top$ along this trajectory. Thus, $\mbf{e} \to 0$ exponentially fast, which in turn implies that $\mbf{p} \to \mbf{p}^*$ for all initial positions outside the set $\mathcal{C}$, where $\mbf{p}^*$ is a point in the desired equilibrium set $\mathcal{P}^*$. Since this result holds for every $\mbf{p}(0) \notin \mathcal{C}$, we conclude that the formation system \eqref{controller} almost globally asymptotically converges to a desired configuration in $\mathcal{P}^*$.
\end{proof}

%%%%%%%%%%%%%%%%%%%%%%%%%%%%%%
\subsection{Simulation}
\begin{figure}[t]
\centering
\subfigure[\scriptsize Trajectories of three agents from initial conditions to final conditions. \normalsize]{
\includegraphics[height = 3.6cm]{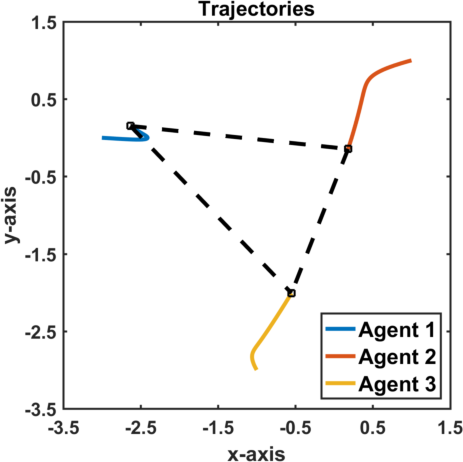}
\label{Fig:simul_control}
}
\subfigure[Two squared distance errors and one cosine error.]{
\includegraphics[height = 3.6cm]{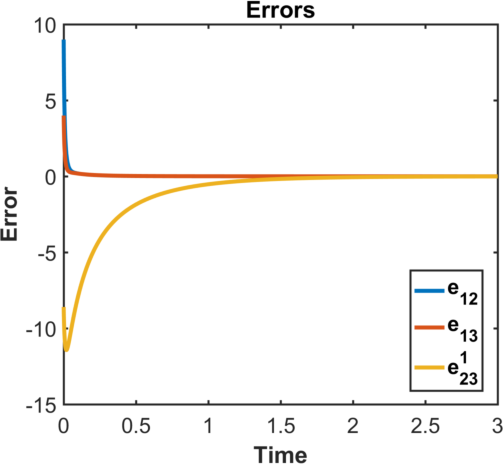}
\label{Fig:simul_error}
}
\caption{Simulation result of 3-agents formation control with 2 distances and 1 angle constraints}
\label{Fig:simul_result}
\end{figure} %

Consider a three-agent system with two distances and one angle constraints as depicted in Fig. \ref{Formation_3e1a}. For the simulation, we set the desired squared relative distances and subtended angle as $d_{12}^{*2}=8$, $d_{13}^{*2}=9$ and $(\theta^{1}_{23})^* = 40^\circ$, and set initial conditions as $\mbf{p}_1(0)=[-3 \; 0]^\top$, $\mbf{p}_2(0)=[1 \; 1]^\top$ and $\mbf{p}_3(0)=[-1 \; -3]^\top$. As a result presented in Fig. \ref{Fig:simul_result}, the squared distance errors and cosine error converge to 0 as time goes by.
%%%%%%%%%%%%%%%%%%%%%%%%%%%%%%%%%%%%%%%%%%%%%%%%%%%%%%%%%%%%%%%%%%%%%%%%%%%%%%%%
\section{Modified Henneberg Construction}
\label{Sec:MHenneberg}
%%%%%%%%%%%%%%%%%%%%%%%%%%%%%%
%\subsection{Modified Henneberg construction}
The Henneberg construction\cite{tay1985generating,eren2004merging} is a technique to grow minimally rigid graphs with the iterative constructions of rigid formations. By using this technique, we define a new technique termed modified Henneberg construction based on the vertex addition and edge splitting of the Henneberg construction. First, we give a definition of minimal weak rigidity.
\begin{definition}[Minimally weakly rigid] \label{Def:minimally_weak_rigidity}
If a framework $(\mathcal{G},\mbf{p})$ is weakly rigid and no single distance- or angle-constraint can be removed without losing the weak rigidity, then the framework is minimally weakly rigid.
\end{definition}
The two operations of the modified Henneberg construction are termed \myemph{weakly rigid 0-extension} and \myemph{weakly rigid 1-extension}, respectively. In the weakly rigid 0-extension, a vertex and two angles are added from the formation illustrated in Fig. \ref{tri_rigid}. Let $\tilde{\mathcal{G}} = (\tilde{\mathcal{V}},\tilde{\mathcal{E}},\tilde{\mathcal{A}})$ be a graph, where a vertex $\nu$ is adjoined so that $\tilde{\mathcal{V}}=\mathcal{V}\cup\{\nu \}$ and $\tilde{\mathcal{A}}=\mathcal{A}\cup\{\theta_{j\nu}^{i}, \theta_{i\nu}^{j} \}$ for some $i,j\in \mathcal{V}$ as illustrated in Fig. \ref{0-extension}. In the weakly rigid 1-extension, a vertex and three angles are added while one existing edge is removed from the formation illustrated in Fig. \ref{tri_rigid}. Let $\tilde{\mathcal{G}} = (\tilde{\mathcal{V}},\tilde{\mathcal{E}},\tilde{\mathcal{A}})$ be a graph, where a vertex $\nu$ is adjoined, while an edge of $\mathcal{G}$ is removed, so that $\tilde{\mathcal{V}}=\mathcal{V}\cup\{\nu \}$, $\tilde{\mathcal{E}}=\mathcal{E} \setminus \{(i,j) \}$ and $\tilde{\mathcal{A}}=\mathcal{A}\cup\{\theta_{j\nu}^{i}, \theta_{i\nu}^{j}, \theta_{ij}^{k} \}$ for some $i,j,k\in \mathcal{V}$ as illustrated in Fig. \ref{1-extension}. From the properties of the constructions, the two operations can be also termed 0-angle splitting and 1-angle splitting, respectively. The modified Henneberg construction can be used to grow minimally rigid (or minimally weakly rigid) formations with additional angles as the following result.

\begin{theorem}
\label{Thm:Mod_Hen_Const}
Frameworks constructed by the weakly rigid 0-extenstion and 1-extension from a framework $(\mathcal{G},\mbf{p})$ are minimally weakly rigid if the framework $(\mathcal{G},\mbf{p})$ is minimally rigid or minimally weakly rigid.
\end{theorem}
\begin{proof}
i$)$ In the case of weakly rigid 0-extension as illustrated in Fig. \ref{0-extension}, the operation is extended from triangular formation as Fig. \ref{tri_rigid}. The three constraints $d_{23}$, $\theta_{34}^{2}$ and $\theta_{24}^{3}$ can be changed to three distance constraints $d_{23}$,  $d_{24}$ and  $d_{34}$ by the law of sines such that
$\frac{d_{23}}{\sin{\theta_{23}^{4}}} = \frac{d_{24}}{\sin{\theta_{24}^{3}}} = \frac{d_{34}}{\sin{\theta_{34}^{2}}}.$
Thus, a formation with three constraints $d_{23}$, $\theta_{34}^{2}$ and $\theta_{24}^{3}$ can be transformed to a formation with three distance constraints  $d_{23}$,  $d_{24}$ and  $d_{34}$, i.e. the formation extended by the weakly rigid 0-extenstion as Fig. \ref{0-extension} can be transformed to the minimally rigid formation as Fig. \ref{rigid_diamond}.
ii$)$ In the case of weakly rigid 1-extension as illustrated in Fig. \ref{1-extension}, the operation is extended from a formation with two edges and subtended angle as in Fig. \ref{Formation_3e1a}. The distance $d_{23}$ can be calculated by the law of cosines as mentioned in Section \ref{Sec:weakRigidity}. Thus, with the proof of the case i) of weakly rigid 0-extension, the formation extended by the weakly rigid 1-extension can be also transformed to the rigid formation as Fig. \ref{rigid_diamond}.
Therefore, if a framework $(\mathcal{G},\mbf{p})$ is minimally rigid or minimally weakly rigid, then frameworks extended by the weakly rigid 0-extenstion or weakly rigid 1-extension are minimally weakly rigid.
\end{proof}

%%%%%%%%%%%%%%%%%%%%%%%%%%%%%%%%%%%%%%%%%%%%%%%%%%%%%%%%%%%%%%%%%%%%%%%%%%%%%%%%
\section{Weak Rigidity in the Three-Dimensional Space}
\label{Sec:weakRigidity_3dim}
%%%%%%%%%%%%%%%%%%%%%%%%%%%%%%
In this section, we extend the weak rigidity in the two-dimensional space to the concept of the three-dimensional space. We do not consider the infinitesimal weak rigidity but just the weak rigidity.

%%%%%%%%%%%%%%%%%%%%%%%%%%%%%%
\subsection{Weak Rigidity from Rigidity Matrix in $\mathbb{R}^{3}$} 

The weak rigidty in $\mathbb{R}^{3}$ can be similarly defined as the weak rigidity in \cite{park2017rigidity}. Consider formations in \myfig\ref{Fig:ExTetrahedron}. The first formation is defined by 3 edge lengths and 3 subtended angles while the second formation is defined by 6 edge lengths. The first formation can be transformed to the second formation with the law of cosines as stated in Section \ref{Sec:weakRigidity}.

\begin{figure}[]
\centering
\subfigure[Non-rigid but weakly rigid framework in $\mathbb{R}^{3}$.]{
\includegraphics[width=0.2\textwidth]{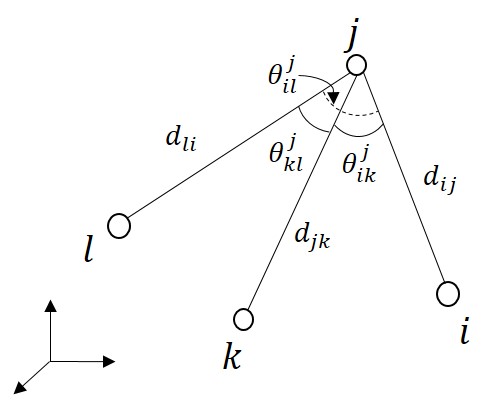}
\label{Fig:ExTetrahedron_a}
}
\subfigure[Rigid framework in $\mathbb{R}^{3}$.]{
\includegraphics[width=0.2\textwidth]{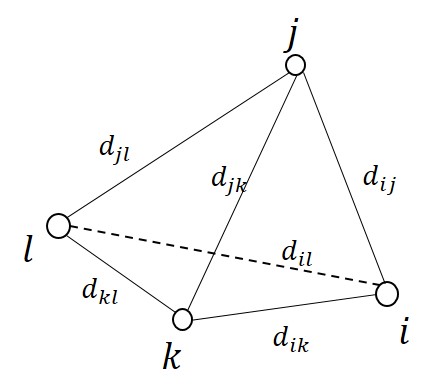}
\label{Fig:ExTetrahedron_b}
}
\caption{Tetrahedral formations under different constraints}
\label{Fig:ExTetrahedron}
\end{figure} %

\begin{definition}[weak rigidity in $\mathbb{R}^{3}$]
\label{Def:weakRigidity}
A framework $(\mathcal{G},\mbf{p})$ is \myemph{weakly rigid} in $\mathbb{R}^{3}$ if there exists a neighborhood $\mathcal{B}_{\mbf{p}} \subseteq \mathbb{R}^{3n}$ of $\mbf{p}$ such that each framework $(\mathcal{G},\mbf{q})$, $\mbf{q} \in \mathcal{B}_{\mbf{p}}$, strongly equivalent to $(\mathcal{G},\mbf{p})$ is congruent to $(\mathcal{G},\mbf{p})$.
\end{definition}

We examine weak rigidity from rigidity matrix. First, the \myemph{rigidity function} $F_D: \mathbb{R}^{dn} \rightarrow \mathbb{R}^{m}$ of $(\mathcal{G},\mbf{p})$ is defined as
$$
F_D(\mbf{p})\equiv[...,\norm{\mbf{p}_{ij}}^2,...]^\top_{(i,j) \in \mathcal{E}} \in 
\mathbb{R}^{m}.
$$
The \myemph{rigidity matrix} then is defined as the Jacobian of the rigidity function:
\begin{equation} \label{rigidity_D_matrix}
R_D(\mbf{p})=\frac{1}{2} \frac{\partial F_D(\mbf{p})}{\partial \mbf{p}} \in 
\mathbb{R}^{m \times dn}.
\end{equation}

\begin{lemma}[\cite{C:Hendrickson:SIAM1992}]
\label{Lemma:rankConditionForIR}
A framework $(\mathcal{G},\mbf{p})$ in $\mathbb{R}^{3}$ with $n \ge 3$ is infinitesimally rigid in $\mathbb{R}^{3}$ if and only if the rank of the rigidity matrix of $(\mathcal{G},\mbf{p})$ is $3n - 6$.
\end{lemma}

Consider a graph $\bar{\mathcal{G}}$, $\bar{\mathcal{G}} = (\bar{\mathcal{V}},\bar{\mathcal{E}},\bar{\mathcal{A}})$, induced from $\mathcal{G}$ in such a way that\cite{park2017rigidity}:
\begin{itemize}
\item \(\bar{\mathcal{V}} = \mathcal{V}\),
\item \(\bar{\mathcal{E}} = \set*{(i,j) \given (i,j) \in \mathcal{E} \lor \exists k \in \mathcal{V}\text{ s.t. } (k,(i,j)) \in \mathcal{A}}\),
\item \(\bar{\mathcal{A}} = \emptyset\).
\end{itemize}
Then, we can obtain the following result: 
\begin{corollary}
\label{COR:WRbyRigidity^3}
A framework $(\mathcal{G},\mbf{p})$ is weakly rigid in $\mathbb{R}^{3}$ if and only if $(\bar{\mathcal{G}},\mbf{p})$ is rigid in $\mathbb{R}^{3}$.
\end{corollary}
\begin{proof}
The proof is similar to the proof of Theorem 1 in \cite{park2017rigidity}. %except that $\mathbb{R}^{2\card{\bar{\mathcal{V}}}}$ and $\mathbb{R}^{2\card{\mathcal{V}}}$ are replaced by $\mathbb{R}^{3\card{\bar{\mathcal{V}}}}$ and $\mathbb{R}^{3\card{\mathcal{V}}}$ respectively in $\mathbb{R}^{3}$.
\end{proof}
With the above result, we know that weak rigidity of \((\mathcal{G},\mbf{p})\) can be determined by the rigidity of \((\bar{\mathcal{G}},\mbf{p})\) indirectly.

The infinitesimal rigidity of a framework is a sufficient condition for the framework to be rigid. The infinitesimal rigidity can be examined by the rank of the rigidity matrix as mentioned in Lemma \ref{Lemma:rankConditionForIR}. Therefore, we can check the weakly rigid of $(\mathcal{G},\mbf{p})$ by investigating the rank of the rigidity matrix of $(\bar{\mathcal{G}},\mbf{p})$ with Corollary \ref{COR:WRbyRigidity^3} and Lemma \ref{Lemma:rankConditionForIR}.
\begin{theorem}
\label{Thm:RImpliesWR}
A framework $(\mathcal{G},\mbf{p})$ with $n \ge 3$ is weakly rigid in $\mathbb{R}^{3}$ if the rigidity matrix $R_D(\mbf{p})$ associated with $(\bar{\mathcal{G}},\mbf{p})$ has rank $3n - 6$.
\end{theorem}
\begin{proof}
If $\rank(R_D(\mbf{p})) = 3n - 6$ associated with $(\bar{\mathcal{G}},\mbf{p})$, then the framework $(\bar{\mathcal{G}},\mbf{p})$ is infinitesimally rigid in $\mathbb{R}^{3}$ by Lemma \ref{Lemma:rankConditionForIR}. Also, the framework $(\bar{\mathcal{G}},\mbf{p})$ is rigid since it is infinitesimally rigid. Therefore, $(\mathcal{G},\mbf{p})$ is weakly rigid in $\mathbb{R}^{3}$ by Corollary \ref{COR:WRbyRigidity^3}.
\end{proof}

A configuration $\mbf{p}$ of a graph is said to be generic if the vertex coordinates are algebraically independent over the rationals \cite{C:Hendrickson:SIAM1992}.

\begin{theorem}[\cite{C:Hendrickson:SIAM1992}]
\label{Thm:Generic_R_IR}
A framework $(\mathcal{G},\mbf{p})$ with generic configuration $\mbf{p}$ is rigid if and only if the framework is infinitesimally rigid.
\end{theorem}
Therefore, if a configuration $\mbf{p}$ of a graph is generic, then we can state the following result.
\begin{corollary}[Generic Property of Graph]
\label{Corol:IRiffWRGenericProp}
If a configuration $\mbf{p}$ is generic, then $(\mathcal{G},\mbf{p})$ with $n \ge 3$ is weakly rigid in $\mathbb{R}^{3}$ if and only if the rigidity matrix $R_D(\mbf{p})$ associated with $(\bar{\mathcal{G}},\mbf{p})$ has rank $3n - 6$.
\end{corollary}
\begin{proof}
Suppose that a given framework $(\mathcal{G},\mbf{p})$ with generic configuration $\mbf{p}$ is rigid. Then, the framework is infinitesimally rigid and vice versa\cite{C:Hendrickson:SIAM1992, C:Asimow:JMAA1979}. Therefore, with Theorems \ref{Thm:RImpliesWR} and \ref{Thm:Generic_R_IR}, if a configuration $\mbf{p}$ is generic, then the rank condition of rigidity matrix $R_D(\mbf{p})$ becomes a necessary and sufficient condition for weak rigidity of $(\mathcal{G},\mbf{p})$.
\end{proof}

%%%%%%%%%%%%%%%%%%%%%%%%%%%%%%
\subsection{Globally Weak Rigidity in $\mathbb{R}^{3}$ } \label{Thm:GWRbyRigidity^3}
We can also extend the local concept of the weak rigidity to the global concept.
With reference to \cite{park2017rigidity}, global weak rigidity can be defined and proved as follows.
\begin{definition}[Global weak rigidity]
\label{Def:globalWeakRigidity}
\qquad \qquad \qquad \qquad A framework $(\mathcal{G},\mbf{p})$ is \myemph{globally weakly rigid} in $\mathbb{R}^{3}$ if any framework $(\mathcal{G},\mbf{q})$, $\mbf{q} \in \mathbb{R}^{3n}$, strongly equivalent to $(\mathcal{G},\mbf{p})$ is congruent to $(\mathcal{G},\mbf{p})$.
\end{definition} %\qquad : a function which inserts a space of 2em in text or math mode;

As Corollary \ref{COR:WRbyRigidity^3} is proved, the following theorem can be also proved easily.
\begin{theorem}
\label{Thm:GWRbyGRigidity} 
A framework $(\mathcal{G},\mbf{p})$ is globally weakly rigid in $\mathbb{R}^{3}$ if and only if $(\bar{\mathcal{G}},\mbf{p})$ is globally rigid in $\mathbb{R}^{3}$. %where $\bar{\mathcal{G}}$ is a graph obtained by substituting $\mathcal{K}(\mathcal{V}_{a,i})$ for $\mathcal{G}_{a,i}$ for all $i \in \set{1,\ldots,c}$.
\end{theorem}
\begin{proof}
The proof is similar to the proof of Corollary \ref{COR:WRbyRigidity^3} except that $\mathcal{B}_{\mbf{p}}$ is replaced by $\mathbb{R}^{3\card{\mathcal{V}}}$.
\end{proof}

%%%%%%%%%%%%%%%%%%%%%%%%%%%%%%%%%%%%%%%%%%%%%%%%%%%%%%%%%%%%%%%%%%%%%%%%%%%%%%%%
\section{CONCLUSIONS}
\label{Sec:conclusion}

We have shown four main results in the paper. First, we introduced the infinitesimal weak rigidity in the two-dimensional space. In the original weak rigidity theory \cite{park2017rigidity}, a framework with constraints of two adjacent edges and a subtended angle must be defined and transformed into a three distance constrained in order to check whether the framework is weakly rigid or not. On the contrary, the infinitesimal weak rigidity of a framework can be directly checked by a rank condition of the weak rigidity matrix associated with the framework. For the infinitesimal weak rigidity, adjacent edges do not need to be defined, that is, a framework with only angle constraints can be also infinitesimally weakly rigid. As the second result, we explored the three-agent formation control using the gradient control law in the two-dimensional space and showed that the formation system exponentially converges to the desired target formation from almost global initial positions. As the third result, we proposed the modified Henneberg construction to build up minimally (weakly) rigid frameworks.
Finally, we extended the weak rigidity in $\mathbb{R}^{2}$ to the concept in $\mathbb{R}^{3}$. The final result shows that (locally) unique formation shape of a framework in $\mathbb{R}^{3}$ can be obtained by the weak rigidity theory even if the framework is not rigid in the viewpoint of (distance) rigidity.

%%%%%%%%%%%%%%%%%%%%%%%%%%%%%%%%%%%%%%%%%%

\addtolength{\textheight}{-12cm}   % This command serves to balance the column lengths
%\section*{APPENDIX}

\section*{ACKNOWLEDGMENT}
The authors would like to thank Zhiyong Sun and Myoung-Chul Park for helpful discussions. 
%%%%%%%%%%%%%%%%%%%%%%%%%%%%%%%%%%%%%%%%%%%%%%%%%%%%%%%%%%%%%%%%%%%%%%%%%%%%%%%%

\bibliographystyle{IEEEtran}
\bibliography{Weak_rigidity_2018}

\end{document}